\documentclass[USenglish]{mylipics}

\usepackage{latexsym}
\usepackage{amsmath}
\usepackage{amssymb}
\usepackage{microtype}
\usepackage{datetime}
\usepackage{wrapfig}
\usepackage{fourier}
\usepackage{mathrsfs}

\usepackage{ifpdf}

\ifpdf
\usepackage[all,arc,curve,color,frame]{xy}
\else
\usepackage[all,dvips,arc,curve,color,frame]{xy}
\fi
\usepackage{color}

\usepackage{algorithmic}

\usepackage{hyperref}

\usepackage{pgf}
\usepackage{tikz}
\usetikzlibrary{arrows,shapes,snakes,backgrounds}

\begin{document}

\theoremstyle{theorem}
\newtheorem{proposition}{Proposition}

\newenvironment{myclaim}
{
	\begin{flushleft}
	\textsf{\textbf{Claim:}}
}
{
	\end{flushleft}

}


\newcommand{\bra}[1]{\left\langle{#1}\right\vert}
\newcommand{\bit}[1]{{#1}}
\newcommand{\qw}[1][-1]{\ar @{-} [0,#1]}
\newcommand{\qwx}[1][-1]{\ar @{-} [#1,0]}
\newcommand{\cw}[1][-1]{\ar @{=} [0,#1]}
\newcommand{\cwx}[1][-1]{\ar @{=} [#1,0]}
\newcommand{\gate}[1]{*{\xy *+<.6em>{#1};p\save+LU;+RU **\dir{-}\restore\save+RU;+RD **\dir{-}\restore\save+RD;+LD **\dir{-}\restore\POS+LD;+LU **\dir{-}\endxy} \qw}
\newcommand{\meter}{\gate{\xy *!<0em,1.1em>h\cir<1.1em>{ur_dr},!U-<0em,.4em>;p+<.5em,.9em> **h\dir{-} \POS <-.6em,.4em> *{},<.6em,-.4em> *{} \endxy}}
\newcommand{\measure}[1]{*+[F-:<.9em>]{#1} \qw}
\newcommand{\measuretab}[1]{*{\xy *+<.6em>{#1};p\save+LU;+RU **\dir{-}\restore\save+RU;+RD **\dir{-}\restore\save+RD;+LD **\dir{-}\restore\save+LD;+LC-<.5em,0em> **\dir{-} \restore\POS+LU;+LC-<.5em,0em> **\dir{-} \endxy} \qw}
\newcommand{\measureD}[1]{*{\xy*+=+<.5em>{\vphantom{\rule{0em}{.1em}#1}}*\cir{r_l};p\save*!R{#1} \restore\save+UC;+UC-<.5em,0em>*!R{\hphantom{#1}}+L **\dir{-} \restore\save+DC;+DC-<.5em,0em>*!R{\hphantom{#1}}+L **\dir{-} \restore\POS+UC-<.5em,0em>*!R{\hphantom{#1}}+L;+DC-<.5em,0em>*!R{\hphantom{#1}}+L **\dir{-} \endxy} \qw}
\newcommand{\multimeasure}[2]{*+<1em,.9em>{\hphantom{#2}} \qw \POS[0,0].[#1,0];p !C *{#2},p \drop\frm<.9em>{-}}
\newcommand{\multimeasureD}[2]{*+<1em,.9em>{\hphantom{#2}}\save[0,0].[#1,0];p\save !C *{#2},p+LU+<0em,0em>;+RU+<-.8em,0em> **\dir{-}\restore\save +LD;+LU **\dir{-}\restore\save +LD;+RD-<.8em,0em> **\dir{-} \restore\save +RD+<0em,.8em>;+RU-<0em,.8em> **\dir{-} \restore \POS !UR*!UR{\cir<.9em>{r_d}};!DR*!DR{\cir<.9em>{d_l}}\restore \qw}
\newcommand{\control}{*!<0em,.025em>-=-{\bullet}}
\newcommand{\controld}{*!<0em,.025em>-=-{\blacktriangledown}}
\newcommand{\controlu}{*!<0em,.025em>-=-{\blacktriangle}}
\newcommand{\controlo}{*-<.21em,.21em>{\xy *=<.59em>!<0em,-.02em>[o][F]{}\POS!C\endxy}}
\newcommand{\ctrl}[1]{\control \qwx[#1] \qw}
\newcommand{\ctrld}[1]{\controld \qwx[#1] \qw}
\newcommand{\ctrlu}[1]{\controlu \qwx[#1] \qw}
\newcommand{\ctrlo}[1]{\controlo \qwx[#1] \qw}
\newcommand{\targ}{*!<0em,.019em>=<.79em,.68em>{\xy {<0em,0em>*{} \ar @{ - } +<.4em,0em> \ar @{ - } -<.4em,0em> \ar @{ - } +<0em,.36em> \ar @{ - } -<0em,.36em>},<0em,-.019em>*+<.8em>\frm{o}\endxy} \qw}
\newcommand{\qswap}{*=<0em>{\times} \qw}
\newcommand{\multigate}[2]{*+<1em,.9em>{\hphantom{#2}} \qw \POS[0,0].[#1,0];p !C *{#2},p \save+LU;+RU **\dir{-}\restore\save+RU;+RD **\dir{-}\restore\save+RD;+LD **\dir{-}\restore\save+LD;+LU **\dir{-}\restore}
\newcommand{\ghost}[1]{*+<1em,.9em>{\hphantom{#1}} \qw}
\newcommand{\push}[1]{*{#1}}
\newcommand{\gategroup}[6]{\POS"#1,#2"."#3,#2"."#1,#4"."#3,#4"!C*+<#5>\frm{#6}}
\newcommand{\rstick}[1]{*!L!<-.5em,0em>=<0em>{#1}}
\newcommand{\lstick}[1]{*!R!<.5em,0em>=<0em>{#1}}
\newcommand{\ustick}[1]{*!D!<0em,-.5em>=<0em>{#1}}
\newcommand{\dstick}[1]{*!U!<0em,.5em>=<0em>{#1}}
\newcommand{\Qcircuit}[1][0em]{\xymatrix @*[o] @*=<#1>}
\newcommand{\node}[2][]{{\begin{array}{c} \ _{#1}\  \\ {#2} \\ \ \end{array}}\drop\frm{o} }
\newcommand{\link}[2]{\ar @{-} [#1,#2]}
\newcommand{\pureghost}[1]{*+<1em,.9em>{\hphantom{#1}}}
   \renewcommand{\Qcircuit}[1][0em]{\xymatrix @*=<#1>}
   
\newcommand{\defref}[1]{Definition~\ref{def:#1}}
\newcommand{\theoref}[1]{Theorem~\ref{theo:#1}}
\newcommand{\propref}[1]{Proposition~\ref{prop:#1}}
\newcommand{\corref}[1]{Corollary~\ref{cor:#1}}
\newcommand{\lemref}[1]{Lemma~\ref{lem:#1}}
\newcommand{\exref}[1]{Example~\ref{ex:#1}}
\newcommand{\reref}[1]{Remark~\ref{re:#1}}
\newcommand{\eref}[1]{(\ref{eq:#1})}
\newcommand{\secref}[1]{Section~\ref{sec:#1}}
\newcommand{\figref}[1]{Fig.~\ref{fig:#1}}
\newcommand{\algoref}[1]{Algorithm~\ref{algo:#1}}

\newcommand{\proves}{\vdash}
\newcommand{\HRule}{\rule{\linewidth}{0.5mm}}
\newcommand{\calf}[1]{\mathcal{#1}}
\newcommand{\set}[1]{\{#1\}}
\newcommand{\bset}[1]{\bigl\{#1\bigr\}}
\newcommand{\It}[1]{\mathsf{#1}}
\newcommand{\MIt}[1]{\mathit{#1}}
\newcommand{\seq}[1]{\langle #1\rangle}
\newcommand{\colq}[1]{\begin{array}{c}#1\end{array}}
\newcommand*\circled[1]{\tikz[baseline=(char.base)]{
            \node[shape=circle,draw,inner sep=2pt] (char) {#1};}}

\newcommand{\fxpset}{\Phi_{\mathsf{fxp}}}
\newcommand{\smset}{\Phi_{\mathsf{sm}}}

\newcommand{\lfmm}{\delta_{\It{LFMM}}}
\newcommand{\ccv}{\delta_{\It{CCV}}}
\newcommand{\dnum}{\delta_{\It{NUM}}}
\newcommand{\conn}{\delta_{\It{CONN}}}
\newcommand{\lfmp}{\delta_{\It{LFMP}}}
\newcommand{\mfv}{\delta_{\It{MFV}}}

\newcommand{\mred}{\le_{\mathsf{m}}}
\newcommand{\ored}{\le_{\mathsf{o}}}

\newcommand{\todo}[1]{ \textcolor{red}{TODO: #1}}
\newcommand{\co}[1]{ \textcolor{red}{#1}}

\newcommand{\M}{\calf{M}}
\newcommand{\Pair}{\mathsf{Pair}}
 \newcommand{\lra}{\leftrightarrow}

\newcommand{\VPV}{\textit{VPV}}
\newcommand{\PV}{\textit{PV}}

\newcommand{\FP}{\mathsf{FP}}
\newcommand{\PH}{\mathsf{PH}}
\newcommand{\CC}{\mathsf{CC}}
\newcommand{\CCall}{\mathsf{CC}_{\mathsf{all}}}
\newcommand{\CCSubr}{\mathsf{CC}^{\mathsf{Subr}}}
\newcommand{\CCstar}{\mathsf{CC}^\ast}
\newcommand{\Class}{\mathsf{C}}
\newcommand{\FCC}{\mathsf{FCC}}
\newcommand{\FCCSubr}{\mathsf{FCC}^{\mathsf{Subr}}}
\newcommand{\FCCstar}{\mathsf{FCC}^\ast}
\newcommand{\FC}{\mathsf{FC}}
\newcommand{\AC}{\mathsf{AC}}
\newcommand{\DET}{\mathsf{DET}}
\newcommand{\RDET}{\mathsf{RDET}}
\newcommand{\NC}{\mathsf{NC}}
\newcommand{\FNC}{\mathsf{FNC}}
\newcommand{\FNL}{\mathsf{FNL}}
\newcommand{\FRNC}{\mathsf{FRNC}}
\newcommand{\NL}{\mathsf{NL}}
\renewcommand{\L}{\mathsf{L}}
\newcommand{\RsL}{\mathsf{R}\#\mathsf{L}}
\newcommand{\TC}{\mathsf{TC}}
\newcommand{\SL}{\mathsf{SL}}
\newcommand{\ZPLP}{\mathsf{ZPLP}}
\newcommand{\ZPL}{\mathsf{ZPL}}
\newcommand{\RP}{\mathsf{RP}}
\newcommand{\RL}{\mathsf{RL}}
\newcommand{\coRP}{\mathsf{coRP}}
\newcommand{\coRL}{\mathsf{coRL}}
\newcommand{\BPL}{\mathsf{BPL}}
\newcommand{\RNC}{\mathsf{RNC}}
\newcommand{\coRNC}{\mathsf{coRNC}}
\newcommand{\RAC}{\mathsf{RAC}}
\renewcommand{\P}{\mathsf{P}}
\newcommand{\Ppoly}{\mathsf{P}/\mathsf{poly}}

\newcommand{\VCC}{\;\mathsf{VCC}^\ast}
\newcommand{\VAC}{\mathsf{VAC}}
\newcommand{\VNC}{\mathsf{VNC}}
\newcommand{\VNL}{\mathsf{VNL}}
\newcommand{\VL}{\mathsf{VL}}
\newcommand{\VP}{\mathsf{VP}}
\newcommand{\VTC}{\mathsf{VTC}}
\newcommand{\VSL}{\mathsf{VSL}}
\newcommand{\VsL}{\mathsf{V\#L}}
\newcommand{\LFP}{\calf{L}_{\FP}}
\newcommand{\VC}{\mathsf{VC}}
\newcommand{\V}{\mathsf{V}}
\newcommand{\CH}{\it{CH}}

\newcommand{\LMAX}{\mathit{LMAX}}
\newcommand{\LMIN}{\mathit{LMIN}}
\newcommand{\LIND}{\mathit{LIND}}
\newcommand{\PIND}{\mathit{PIND}}
\newcommand{\IND}{\textit{IND}}
\newcommand{\PHP}{\mathsf{PHP}}
\newcommand{\NUMO}{\textit{NUMONES}}
\newcommand{\MFV}{\textit{MFV}}
\newcommand{\MCV}{\textit{MCV}}
\newcommand{\CONN}{\textit{CONN}}
\newcommand{\CCVA}{\textit{CCV}}
\newcommand{\CCVB}{\It{CCV}}

\newcommand{\SMP}{ \textmd{\textsc{Sm}}}
\newcommand{\MOSM}{ \textmd{\textsc{MoSm}}}
\newcommand{\WOSM}{ \textmd{\textsc{WoSm}}}
\newcommand{\BSVP}{ \textmd{\textsc{Bsvp}}}
\newcommand{\LFMM}{ \textmd{\textsc{Lfmm}}}
\newcommand{\TLFMM}{ \textmd{\textsc{3Lfmm}}}
\newcommand{\CCV}{ \textmd{\textsc{Ccv}}}
\newcommand{\MCVP}{ \textmd{\textsc{Mcvp}}}
\newcommand{\XNS}{ \textmd{\textsc{Xns}}}
\newcommand{\TCV}{ \textmd{\textsc{Three-valued Ccv}}}

\newcommand{\CCVN}{ \textmd{\textsc{Ccv}$\neg$}}
\newcommand{\VLFMM}{ \textmd{\textsc{vLfmm}}}
\newcommand{\TVLFMM}{ \textmd{\textsc{3vLfmm}}}

\numberwithin{equation}{section}

\title{Complexity Classes and Theories for the Comparator Circuit Value
Problem}
\author[1]{Stephen A. Cook}
\author[1]{Dai Tri Man L\^e}
\author[1]{Yuli Ye}
\authorrunning{S.A. Cook, D.T.M. L\^e, and Y. Ye}
\affil[1]{Department of Computer Science, University of Toronto}



\maketitle
\begin{center}
\fbox{Date: \textit{\today}}  
\end{center}

\begin{abstract} 
Subramanian defined the complexity class $\CC$ as the set of problems
log-space reducible
to the comparator circuit value problem.  He proved that several
other problems are complete for $\CC$, including the stable
marriage problem, and finding the lexicographical first maximal
matching in a bipartite graph.  We suggest alternative definitions
of $\CC$ based on different reducibilities and introduce
a two-sorted theory $\VCC$ based on one of them.  We sharpen and simplify 
Subramanian's completeness proofs for the above two problems and
show how to formalize them in $\VCC$.
\end{abstract}

\tableofcontents

\section{Introduction}\label{s:intro}
Comparator networks were originally introduced as a method of sorting
numbers (as in Batcher's even-odd merge sort \cite{Bat68}), but they are still
interesting when the numbers are restricted to the Boolean values
$\{0,1\}$.  A comparator gate has two inputs $p,q$ and two outputs
$p',q'$, where $p' = \min\{p,q\}$ and $q'=\max\{p,q\}$.  In the
Boolean case (which is the one we consider) $p' = p\wedge q$ and
$q' = p \vee q$.  A comparator circuit (i.e. network) is presented as
a set of $m$ horizontal lines in which the $m$ inputs are presented at the
left ends of the lines and the $m$ outputs are presented at the
right ends of the lines, and in between there is a sequence of
comparator gates, each represented as a vertical arrow connecting
some wire $w_{i}$ with some wire $w_{j}$ as shown in \figref{f0}.  These arrows
divide each wire into segments, each of which gets a Boolean value.
The values of wires $w_i$ and $w_{j}$ after the arrow are the comparator
outputs of the values of wires $w_{i}$ and $w_{j}$ right before the arrow, with the tip of the arrow
representing the maximum.

\begin{wrapfigure}{r}{0.45\textwidth}
  \begin{center}
    $\footnotesize \Qcircuit @C=1.5em @R=0.5em {
\push{\bit{1}}&\push{w_{0}}	& \ctrl{2}	&\push{\,0\,}\qw& \qw 	&\qw
& \ctrl{3}&\push{\,0\,}\qw 	& \qw&   \rstick{\bit{0}} \qw \\
\push{\bit{1}}&\push{w_{1}}	&\qw		&\qw		&\ctrl{2} 	&\push{\,0\,}\qw
&\qw&\qw  	&\ctrlu{1}&   \rstick{\bit{1}} \qw\\
\push{\bit{1}}&\push{w_{2}}	&\qw		&\qw		&\qw  	&\qw
&\qw&\qw  	&\qw &   \rstick{\bit{1}} \qw\\
\push{\bit{0}}&\push{w_{3}}	& \ctrld{-1} &\push{1}\qw& \qw  &\qw
&\qw& \qw 	&\ctrl{-1}&    \rstick{\bit{0}} \qw \\
\push{\bit{0}}&\push{w_{4}}	& \qw 	&\qw		& \ctrld{-1} &\push{\,1\,}\qw
&\qw& \qw & \qw&    \rstick{\bit{1}} \qw \\
\push{\bit{0}}&\push{w_{5}}	& \qw 	&\qw		& \qw 	&\qw
& \ctrld{-3} &\push{\,0\,}\qw& \qw&    \rstick{\bit{0}} \qw 
}$
  \end{center}
  \vspace{-3mm}
  \caption{}
  \label{fig:f0}
\end{wrapfigure}

The comparator circuit value problem ($\CCV$) is:  given a comparator
circuit with specified Boolean inputs, determine the output value of a
designated wire.  To turn this into a complexity class it seems
natural to use a reducibility notion that is weak but fairly robust. 
Thus we define $\CC$ to consist of those problems (uniform) $\AC^0$
many-one-reducible
to $\CCV$.  However Subramanian \cite{Sub94} studied the complexity of
$\CCV$ using a stronger notion of reducibility.  Thus his class,
which we denote $\CCSubr$, consists of those problems log-space
(many-one)-reducible to $\CCV$.  It turns out that a generalization
of  many-one $\AC^0$-reducibility which we will call
{\em oracle $\AC^0$-reducibility} (called simply $\AC^0$ reducibility
in \cite{CN10}), is also useful.  Standard complexity
classes such as $\AC^0$,  $\L$ (log space), $\NL$ (nondeterministic
log space), $\NC$, and $\P$ are all closed under this $\AC^0$  oracle
reducibility.
We denote the closure of $\CCV$ under this reducibility by $\CCstar$.

We will show that
\begin{equation}\label{e:CCs}
    \NL \subseteq \CC\subseteq \CCSubr \subseteq \CCstar \subseteq \P
\end{equation}
The last inclusion is obvious because $\CCV$ is a special case
of the monotone circuit value problem, which is clearly in $\P$.
The inclusion $\CC\subseteq \CCSubr$ follows because $\AC^0$
functions are also log-space functions.  The inclusion
$\CCSubr \subseteq \CCstar$ follows from the first inclusion,
which in turn is a strengthening of a result in \cite{MS92}
(attributed to Feder) showing that $\NL \subseteq \CCSubr$.
Of course all three comparator classes coincide if it turns out
that $\CC$ is closed under oracle $\AC^0$-reductions, but we
do not know how to show this.

Note that comparator circuits
are more restricted than monotone Boolean circuits because each
comparator output has fan-out one.  This leads to the 
open question of whether $\CCstar \subsetneq \P$.   A second
open question is whether the complexity classes $\CCstar$
and $\NC$ are incomparable.  (Here $\NC$ is the class of problems
computed by uniform circuit families of polynomial size and polylog
depth, and satisfies $\NL\subseteq \NC \subseteq \P$.)
The answers could be different
if we replaced $\CCstar$ by $\CCSubr$ or $\CC$, although
$\CC \subseteq \NC$ iff $\CCstar \subseteq \NC$ because $\NC$
is closed under oracle $\AC^0$reductions.

The above classes associated with $\CCV$ are also interesting because
they have several disparate complete problems.  As shown in
\cite{MS92,Sub94} both the lexicographical first maximal
matching problem ($\LFMM$) and the stable marriage problem ($\SMP$)
are complete for $\CCSubr$ under log-space reductions\footnote{The  
second author outlined a proof that $\LFMM$ is complete under $\NC^1$ 
reductions in unpublished notes from 1983.}.
The $\SMP$ problem is especially interesting:  Introduced by Gale
and Shapley in 1962 \cite{GS62}, it has since been used to pair medical
interns with hospital residencies jobs in the USA. 
$\SMP$ can be stated as follows:  Given
$n$ men and $n$ women, each with a complete ranking according to
preference of all $n$ members of the opposite sex, find a complete
matching of the men and women such that there are no two people of
opposite sex who would both rather have each other than their current
partners.  Gale and Shapley proved that such a `stable' matching always
exists, although it may not be unique.  Subramanian \cite{Sub94}
showed that
$\SMP$ treated as a search problem (i.e. find any stable marriage)
is complete for $\CC$ under log-space reducibility. 

Strangely the $\CC$ classes have received very little attention
since Subramanian's papers \cite{Sub90,Sub94}.  The present paper
contributes to their complexity theory by sharpening these
early results and simplifying their proofs. 
For example we prove that the three problems $\CCV$, $\LFMM$, and
$\SMP$ are inter-reducible under many-one $\AC^0$-reductions as
opposed to log-space reductions.  Also we introduce a three-valued
logic version of $\CCV$ to facilitate its reduction to $\SMP$.
Our paper contributes to the proof complexity of the classes
by introducing a two-sorted formal theory $\VCC$
which captures the class $\CCstar$ and which can formalize the proofs
of the above results. 


Our theory $\VCC$ is a two-sorted theory developed in the
way described in \cite[Chapter 9]{CN10}.  In general this
method associates a theory $\VC$ with a suitable complexity class
$\Class$ in such a way that a function is in $\FC$, the function
class associated with $\Class$, iff it is provably total
in $\VC$.  (A string-valued function is in $\FC$ iff it is
polynomially bounded and its bit-graph is in $\Class$.) 
This poses a problem for us
because the provably-total functions in a theory are always closed under
composition, but it is quite possible that neither of the function classes
$\FCC$ and $\FCCSubr$ is closed under composition.  That is why
we define the class $\CCstar$ to consist of the problems
$\AC^0$-oracle-reducible (see Definition \ref{d:ACred} below)
to $\CCV$, rather than the problems $\AC^0$ many-one reducible
to $\CCV$, which comprise $\CC$.
It is easy to see that the functions in $\FCCstar$ are closed
under composition, and these are the functions that are
provably total in our theory $\VCC$.

The above paragraph illustrates one way that studying proof complexity can
contribute to main-stream complexity theory, namely by mandating
the introduction of
the more robust version $\CCstar$ of $\CC$ and $\CCSubr$.
Another way is by using the simple two-sorted syntax of our theories
to demonstrate $\AC^0$ reductions.  Thus Theorem 1 below states
that a simple syntactic class of formulas represents precisely
the $\AC^0$ relations.  In general it is much easier to write down
an appropriate such formula than to describe a uniform circuit family or
alternating Turing machine program.

Once we describe our theory $\VCC$ in Sections \ref{s:two-sorted}
and \ref{s:two-theories}, the technical part of our proofs
involve high-level descriptions of comparator circuits and
algorithms.   We do not say much about formalizing the proofs
in $\VCC$ since in most cases this part is relatively straightforward.
However there are some cases (Theorems \ref{theo:NL2CCV} and 
\ref{theo:fxp}) where the proofs require a counting argument,
and for these we use the fact (Corollary \ref{cor:VTC}) that the
counting theory $\VTC^0$ is contained in $\VCC$.

\section{Preliminaries}\label{s:prelim}

\subsection{Two-sorted vocabularies} \label{s:two-sorted}
We use two-sorted vocabularies for our theories as
described by Cook and Nguyen \cite{CN10}.
Two-sorted languages have variables $x,y,z,\ldots$ ranging over 
$\mathbb{N}$ and variables $X,Y,Z,\ldots$ ranging over finite subsets
of $\mathbb{N}$, interpreted 
as bit strings. Two sorted vocabulary $\mathcal{L}_{A}^{2}$ includes the usual symbols $0,1,+,
\cdot,=,\le$ for arithmetic over $\mathbb{N}$, the length function $|X|$ for strings ($|X|$ is zero if $X$ is empty, otherwise $1+\max(X)$), the set 
membership relation $\in$, and string equality $=_{2}$ (subscript 2 is usually omitted). 
We will use the notation $X(t)$ for $t\in X$, and think of $X(t)$ as the $t^{\text{th}}$ bit in the string $X$.

The number terms in the base language $\mathcal{L}_{A}^{2}$  are built
from the constants $0,1$, variables $x,y,z,
\ldots$ and length terms $|X|$ using $+$ and $\cdot$. The only string terms are string variables, but 
when we extend $\mathcal{L}_{A}^{2}$ by adding string-valued functions, other string terms will be 
built as usual. The atomic formulas are $t=u$, $X=Y$, $t\le u$, $t\in X$ for any number terms $x,y$ 
and string variables $X,Y$. Formulas are built from atomic formulas using $\wedge,\vee,\neg$ and 
$\exists x$, $\exists X$, $\forall x$, $\forall X$. Bounded number quantifiers are defined as usual, 
and bounded string quantifier $\exists X\le t, \varphi$ stands for $\exists X (|X|\le t\, \wedge \varphi)$ 
and $\forall X\le t, \varphi$ stands for $\forall X (|X|\le t \rightarrow \varphi)$, where $X$ does not 
appear in term $t$.

The class $\Sigma_{0}^{B}$ consists of all $\mathcal{L}_{A}^{2}$-formulas with  no string 
quantifiers and only bounded number quantifiers.  The class $\Sigma_{1}^{B}$ consists of 
formulas  of the form 
$\exists \vec{X}< \vec{t}\,\varphi$, where $\varphi\in \Sigma_{0}^{B}$ and the prefix of the bounded 
quantifiers might be empty. These classes are extended to $\Sigma_{i}^{B}$ (and $\Pi_{i}^{B}$)
for all $i\ge 0$, in the usual way.
More generally we write $\Sigma_{i}^{B}(\mathcal{L})$ to denote the
class of  $\Sigma_{i}^{B}$-formulas which may have function and
predicate symbols from $\mathcal{L}\cup \mathcal{L}_{A}^{2}$.

Two-sorted complexity classes contain relations $R(\vec{x},\vec{X})$, where $\vec{x}$ are number 
arguments and $\vec{X}$ are string arguments.  (In the sequel we refer
to a relation $R(\vec{x},\vec{X})$ as a {\em decision problem}: given
$(\vec{x},\vec{X})$ determine whether $R(\vec{x},\vec{X})$ holds.)
In defining complexity classes using machines or 
circuits, the number arguments are represented in unary notation and the string arguments are 
represented in binary. The string arguments are the main inputs, and the number arguments are 
auxiliary inputs that can be used to index the bits of strings.
Using these input conventions, we define the two-sorted version of
$\AC^{0}$ to be the class of relations $R(\vec{x},\vec{X})$ 
such that some alternating Turing machine accepts $R$ in time $O(\log n)$ with a constant 
number of alternations. Then the descriptive complexity characterization of $\AC^{0}$ gives rise to the following theorem \cite[Chapter 4]{CN10}.
\begin{theorem}\label{theo:szb}
A relation $R(\vec{x},\vec{X})$ is in $\AC^{0}$ iff it is represented by some $\Sigma_
{0}^{B}$-formula $\varphi(\vec{x},\vec{X})$.
\end{theorem}

Given a class of relations $\Class$, we associate a class $\FC$ of
string-valued functions $F (\vec{x},\vec{X})$ and number functions
$f(\vec{x},\vec{X})$ with $\Class$ as follows. We require
these functions to be $p$-bounded, i.e., $|F(\vec{x},\vec{X})|$
and $f(\vec{x},\vec{X})$ are bounded by
a polynomial in $\vec{x}$ and $|\vec{X}|$. Then we define $\FC$
to consist of all $p$-bounded number
functions whose graphs are in $\Class$ and all $p$-bounded string
functions whose bit graphs are in $\Class$.  (Here the {\em bit graph}
of $F(\vec{x}, \vec{X})$ is the relation $B_F(i,\vec{x},\vec{X})$
which holds iff the $i$th bit of $F(\vec{x},\vec{X})$ is 1.)

Most of the computational problems we consider here can be nicely
expressed as decision problems (i.e. relations), but the stable
marriage problem $\SMP$ is an exception, because in general a given
instance has more than one solution (i.e. there is more than one stable
marriage).
Thus $\SMP$ is properly described as a search problem.  Following
\cite[Section 8.5]{CN10}  we define a two-sorted
{\em search problem} $Q_R$ to be a multivalued function with graph
$R(\vec{x},\vec{X},Z)$, so
\begin{equation}\label{eq:search}
   Q_R(\vec{x},\vec{X}) = \bset{Z\mid R(\vec{x},\vec{X},Z)}
\end{equation}
The search problem is {\em total} if the set $Q_R(\vec{x},\vec{X})$
is non-empty for all $\vec{x},\vec{X}$.  The search problem is a
{\em function problem} if $|Q_R(\vec{x},\vec{X})|=1$  for all
$\vec{x},\vec{X}$.
A function $F(\vec{x},\vec{X})$ {\em solves} $Q_R$ if
$$
    F(\vec{x},\vec{X})\in Q_R(\vec{x},\vec{X})
$$
for all $\vec{x},\vec{X}$.

Here we will be concerned only with total search problems.

\begin{definition}\label{d:manyOne}
Let $\Class$ be a complexity class.
A relation $R_1(\vec{x},\vec{X})$ is $\Class$-many-one-reducible to
a relation $R_2(\vec{y},\vec{Y})$ (written
$R_1\mred^\Class R_2) $ if there are functions
$\vec{f},\vec{F}$ in $\FC$ such that
\[
R_1(\vec{x},\vec{X}) \lra
R_2(\vec{f}(\vec{x},\vec{X}),\vec{F}(\vec{x},\vec{X})).
\]
A search problem $Q_{R_1}(\vec{x},\vec{X})$ is $\Class$-many-one-reducible to
a search problem $Q_{R_2}(\vec{y},\vec{Y})$ if there are functions
$G,\vec{f},\vec{F}$ in $\FC$ such that
\[
\mbox{$G(\vec{x},\vec{X},Z) \in Q_{R_1}(\vec{x},\vec{X})$ for all
$Z\in Q_{R_2}(\vec{f}(\vec{x},\vec{X}),\vec{F}(\vec{x},\vec{X}))$}.
\]
\end{definition}

Here we are mainly interested in the cases that $\Class$ is
either $\AC^0$ or $\L$ (log space).  We also need a generalization
of many-one $\AC^0$-reducibility called simply $\AC^0$-reducibility
in \cite[Definition 9.1.1]{CN10}, which we will call 
oracle $\AC^0$-reducibility.  Roughly speaking a function or
relation is $\AC^0$-oracle-reducible to a collection ${\mathcal L}$ of
functions and relations if it can be computed by a uniform polynomial
size constant depth family of circuits which have unbounded fan-in
gates computing functions and relations from ${\mathcal L}$ (i.e.
`oracle gates'), in addition to Boolean gates.  Formally:

\begin{definition}\label{d:ACred}
A string function $F$ is $\AC^0$-oracle-reducible
to a collection  ${\mathcal L}$ of relations and functions
(written $F \ored^{\AC^0}{\mathcal L}$)
if there is a sequence of string functions
$F_1,\ldots,F_n=F$ such that each $F_i$ is $p$-bounded and its
bit graph is represented by a $\Sigma^B_0(\mathcal{L},F_1,\ldots,F_{i-1})$-formula.  

A number function $f$ is
$\AC^0$-oracle-reducible to $\mathcal{L}$ if $f = |F|$ for
some string function
$F$ which is $\AC^0$-reducible to $\mathcal{L}$.  A relation $R$
is $\AC^0$-oracle-reducible to $\mathcal{L}$ if its characteristic function
is $\AC^0$-oracle-reducible to $\mathcal{L}$.
\end{definition}

We note that standard small complexity classes including
$\AC^{0}$, $\TC^0$, $\NC^{1}$, $\NL$ and $\P$
(as well as their corresponding
function classes) are closed under oracle $\AC^0$-reductions.


\subsection{Two-sorted theories}\label{s:two-theories}

The theory $\V^{0}$ for $\AC^{0}$ is the basis for developing theories
for small complexity classes 
within $\P$ in \cite{CN10}.  $\V^{0}$ has the vocabulary
$\mathcal{L}_{A}^{2}$ and is axiomatized by the set of
\textit{2-BASIC} axioms as given in \figref{2basic}, which express basic properties of symbols 
in $\mathcal{L}_{A}^{2}$, together with the \emph{comprehension} axiom schema 
\begin{align*}
\Sigma_{0}^{B}\textit{-COMP}: &&\exists X\le y\, \forall z<y \bigl(X(z) 
\leftrightarrow \varphi(z)\bigr),
\end{align*}
where $\varphi \in \Sigma_{0}^{B}(\mathcal{L}_{A}^{2})$ and $X$ does not occur free in $\varphi$. 
Note that the axioms B1--B12 of \textit{2-BASIC} are based on the \textit{1-BASIC} axioms of theory $I\Delta_{0}$; the axioms L1 and L2 characterize $|X|$.

Although  $\V^{0}$ has no explicit induction axiom, nevertheless, using $\Sigma_{0}^{B}\textit{-COMP}$ and the fact that $|X|$ produces the maximum element of the finite set $X$,
the following schemes are provable in $\V^{0}$ for every formula $\varphi \in \Sigma_{0}^{B}(\mathcal{L}_{A}^{2})$
\begin{align*}
\Sigma_{0}^{B}\textit{-IND}: &&  \bigl[\varphi(0)\wedge \forall x \bigl(\varphi(x)\rightarrow \varphi (x+1)\bigr)\bigr]\rightarrow \forall x \varphi(x),\\
\Sigma_{0}^{B}\textit{-MIN}: && \varphi(y)\rightarrow \exists x \bigl(\varphi(x)\wedge \neg 
\exists z<x\, \varphi(z)\bigr),\\
\Sigma_{0}^{B}\textit{-MAX}:&&\varphi(0)\rightarrow \exists x \le y \bigl[\varphi(x)\wedge \neg \exists z\le y \bigl(z>x \wedge  \varphi(z)\bigr)\bigr].
\end{align*}

\begin{figure}
\HRule
\begin{center}
\begin{minipage}{.45\textwidth}
\begin{align*}
&\textbf{B1. }x+1 \not= 0						\\
&\textbf{B2. }x+1=y+1 \rightarrow x= y			\\
&\textbf{B3. }x+0=x							\\
&\textbf{B4. }x+(y+1)=(x+y)+1					\\
&\textbf{B5. }x\cdot 0 =0						\\
&\textbf{B6. }x\cdot (y+1)  = (x\cdot y)+x			\\
&\textbf{B7. }(x\le y \wedge y\le x)\rightarrow x = y\\
\end{align*}
\end{minipage}
\begin{minipage}{.45\textwidth}
\begin{align*}
&\textbf{B8. }x\le x + y\\
&\textbf{B9. }0\le x\\
&\textbf{B10. }x\le y \vee y\le x\\
&\textbf{B11. }x\le y \leftrightarrow x < y +1\\
&\textbf{B12. }x\not= 0 \rightarrow \exists y\le x \left(y+1=x\right)\\
&\textbf{L1. }X(y)\rightarrow y<|X|				\\
&\textbf{L2. }y+1 = |X| \rightarrow X(y)\\
\end{align*}
\end{minipage}
$\textbf{SE. }\bigl(|X|=|Y| \wedge \forall i <|X|\left(X(i)=Y(i)\right)\bigr) \rightarrow X = Y$
\end{center}
\HRule
\caption{The \textit{2-BASIC} axioms}
\label{fig:2basic}
\end{figure}

In general, we say that a string function $F(\vec{x},\vec{X})$ is
$\Sigma^B_1$-definable (or provably total) in a two-sorted theory
$\mathcal{T}$ if there is a $\Sigma^B_1$ formula
$\varphi(\vec{x},\vec{X},Y)$ representing the graph $Y=F(\vec{x},\vec{X})$
of $F$ such that
$\mathcal{T} \vdash \forall\vec{x}\,\forall\vec{X}\exists!Y\varphi(\vec{x},\vec{X},Y)$.  Similarly for a number function $f(\vec{x},\vec{X})$.

It was shown in \cite[Chapter 5]{CN10} that $\V^{0}$ is finitely
axiomatizable, and  a function is provably total in $\V^0$ iff it
is in $\mathsf{FAC}^{0}$.

In  \cite[Chapter 9]{CN10}, Cook and Nguyen develop a general method
for associating a theory $\VC$ with certain 
complexity classes $\mathsf{C}\subseteq \P$, where $\VC$ extends
$\V^{0}$ with an additional axiom 
asserting the existence of a solution to a complete problem for
$\mathsf{C}$.  In order for this method to work, the class $\mathsf{C}$
must be closed under $\AC^0$-oracle-reducibility
(Definition \ref{d:ACred}).
The method shows how to define a universal
conservative extension $\overline{\VC}$ of $\VC$,  where $\overline{\VC}$
has string function symbols for precisely the string functions of
$\mathsf{FC}$,
and terms for precisely the number functions of $\mathsf{FC}$.  Further, 
$\overline{\VC}$ proves the $\Sigma_{0}^{B}(\mathcal{L})\textit{-IND}$
and $\Sigma_{0}^{B}(\mathcal{L})\textit{-MIN}$ schemes, where
$\mathcal{L}$ is the vocabulary of $\overline{VC}$.  It follows from
the Herbrand Theorem that the provably total functions of both $\VC$
and $\overline{\VC}$ are precisely those in $\mathsf{FC}$.  

Using this framework Cook and Nguyen define specific theories for
several complexity classes and give examples of theorems formalizable in
each theory.  The theories of
interest to us in this paper are $\VTC^{0}$, $\VNC^{1}$, $\VNL$ and $\VP$ for the complexity classes $\TC^{0}$, $\NC^{1}$, $\NL$ and $\P$ respectively. All of these theories have vocabulary $\mathcal{L}_{A}^{2}$. Let $\seq{x,y}$ denote the \emph{pairing function}, which is the $\mathcal{L}_{A}^{2}$ term $(x+y)(x+y+1)+2y$. The theory $\VTC^{0}$  is axiomatized by the axioms of $\V^{0}$ and the axiom:
\begin{align}
\NUMO: && \exists Z \le 1 + \seq{n,n}, \dnum(n,X,Z),	\label{eq:numones}
\end{align}
where the formula $\dnum(n,X,Z)$ asserts that $Z$ is a matrix consisting of $n$ rows such that for every $y\le n$, the $y^{\rm th}$ row  of $Z$ encodes the number of ones in the prefix of length $y$ of the binary string $X$. Thus,  the $n^{\rm th}$ row of $Z$ essentially ``counts'' the number of ones in  $X$. Because of this counting ability, $\VTC^{0}$ can prove \emph{the pigeonhole principle} $\PHP(n,F)$ saying that if $F$ maps a set of $n+1$ elements to a set of $n$ elements, then $F$ is not an injection.

The theory $\VNC^{1}$  is axiomatized by the axioms of $\V^{0}$ and the axiom:
\begin{align}
\MFV: && \exists Y \le 2n+1, \mfv(n,F,I,Y),	\label{eq:mfv}
\end{align}
where $F$ and $I$ encode a monotone Boolean formula with $n$ literals and its input respectively, and the formula $\mfv(n,G,I,Y)$ holds iff $Y$ correctly encodes the evaluation of the formula encoded in $F$ on input $I$.  Recall that the \emph{monotone Boolean formula value} problem is complete for $\NC^{1}$.

The theory $\VP$  is axiomatized by the axioms of $\V^{0}$ and the axiom $\MCV$, which is defined very similarly to $\MFV$, but the \emph{monotone circuit value} problem is used instead.

The theory $\VNL$  is axiomatized by the axioms of $\V^{0}$ and the axiom:
\begin{align}
\CONN: && \exists U\le \seq{n,n}+1, \conn(n,E,U),	\label{eq:conn}
\end{align}
where $E$ encodes the edge relation of a directed graph $G$ with $n$
vertices $v_{0},\ldots,v_{n-1}$, and the formula  $\conn(n,E,U)$ is
intended to mean that $U$ is a matrix of $n$ rows, where each row has
a Boolean value for each vertex in $G$,  and $U(d,i)$ holds iff vertex
$v_i$ has distance at most $d$ from $v_0$.  More directly, the
formula $\conn(n,E,U)$ asserts
\begin{align}\label{eq:connMean}
&&\begin{array}{c}
\mbox{$U(0,i)$ holds iff $i=0$, and $U(d+1,i)$ holds iff
either $U(d,i)$ holds or} \\
\mbox{there is $j$ such that $U(d,j)$ holds and there is an edge in $G$
from $v_j$ to $v_i$.} 
\end{array}
\end{align}

Similar to what is currently known about complexity classes, it was shown in \cite[Chapter 9]{CN10} that
the following inclusions hold:
\begin{equation}\label{eq:theories}
\V^{0}\subsetneq \VTC^{0}\subseteq \VNC^{1} \subseteq \VNL \subseteq \VP.
\end{equation}

\subsection{$\CCV$ and its complexity classes}
A \emph{comparator gate} is a function $C:\set{0,1}^{2}\rightarrow \set{0,1}^{2}$, that takes an input  pair $(p,q)$ and outputs a pair $(p\wedge q,p \vee q)$. Intuitively, the first output in the pair is the smaller bit among the two input bits $p,q$, and the second output is the larger bit. \smallskip\\
\begin{minipage}{.65\textwidth}
\hspace{.5cm} We will use the graphical notation on the right to denote a comparator gate, where $x$ and $y$ denote the names of the wires, and the direction of the arrow denotes the direction to which we move the larger bit as shown in the picture. \smallskip
\end{minipage}
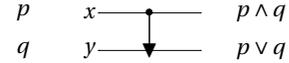
\begin{minipage}{.3\textwidth}
\vspace{-5mm}
\begin{center}\small 
$\Qcircuit @C=2.5em @R=1em {
\push{\bit{p}}&\push{x\,} & \ctrl{1}	& \rstick{~~~~~\bit{p\wedge q}}\qw\\
\push{\bit{q}}&\push{y\,} & \ctrld{-1} & \rstick{~~~~~\bit{p\vee q}}\qw  
}$
\end{center}
\end{minipage}

A \emph{comparator circuit} is a directed acyclic  graph consisting of: \emph{input nodes}  with in-degree zero and out-degree one,  \emph{output nodes} with in-degree one and out-degree zero, and \emph{internal nodes} with  in-degree two and out-degree two, where each internal node is labelled with a comparator gate. We also require each output computed by a comparator gate has fan-out exactly one. Under this definition, each comparator circuit can be seen as consisting of the wires that  carry the bit values  and are arranged in parallel, and each comparator gate connects exactly two wires as previously shown  in \figref{f0}.

The \emph{comparator circuit value} problem ($\CCV$) is the decision
problem:  Given a comparator circuit, an input assignment, and a
designated wire, decide whether the circuit outputs one on that wire.

Formally $\CCV$ is a two-sorted relation $R(\vec{x},\vec{X})$ like
those discussed in Section \ref{s:prelim}.  An instance of the
decision problem is encoded in some reasonable way by the variables
$(\vec{x},\vec{X})$.  The exact encoding is not important.  
An example encoding is $(n,m,i,I,G)$ where $n$, $m$ and $i$ are number
variables and $I$ and $G$ are string variables.  Here $n$ is
the number of wires in the comparator circuit, $m$ is the number of comparator gates,
$i$ is the designated wire, $I$ specifies the sequence of input values to the
wires, and $G$ specifies the sequence of $m$ comparator gates in the
circuit.

\begin{definition}\label{d:CC}
The complexity class $\CC$ (resp. $\CCSubr$, $\CCstar$) is the class of
decision problems (i.e. relations) that are $\AC^0$ many-one-reducible
(resp. log space-reducible, $\AC^0$ oracle-reducible)
to $\CCV$. 
A decision problem $R$ is $\CC$-complete (resp. $\CCSubr$-complete,
$\CCstar$-complete) if the respective class is the closure of $R$
under the corresponding reducibility.  We say that $R$ is
$\CCall$-complete if it is complete in all three senses.
\end{definition}


The next result is a straightforward consequence of (\ref{e:CCs})
and the definitions involved.

\begin{lemma}\label{l:CCall}
If a decision problem is $\CC$-complete then it is $\CCall$-complete.
\end{lemma}

In the above definition of comparator circuit, each comparator gate can point in either direction, upward or downward (see \figref{f0}). However, it is not hard to show the following. 

\begin{proposition}\label{prop:p1}
The $\CCV$ problem with the restriction that all comparator gates point in the same direction is $\CC$-complete.
\end{proposition}
\begin{proof}Suppose we have a gate with the arrow pointing upward like following:
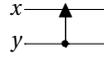
\begin{center}\small 
$\Qcircuit @C=2em @R=1em {
\push{x\,} & \ctrlu{1}	& \qw\\
\push{y\,} & \ctrl{-1} & \qw  
}$
\end{center}
We will build a circuit that outputs the same values as those of $x$ and $y$, but all the gates will now point downward.
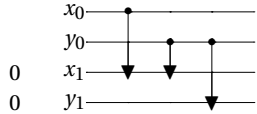
\begin{center}\small
$\Qcircuit @C=2em @R=0.7em {
		&\push{x_{0}} & \ctrl{2}		& \qw	& \qw 		& \qw\\
		&\push{y_{0}} &\qw		& \ctrl{1} 	& \ctrl{2}  		& \qw\\
\push{0}	&\push{x_{1}} & \ctrld{-1} 	&\ctrld{-1}	&\qw			& \qw\\
\push{0}	&\push{y_{1}} & \qw		& \qw 	&\ctrld{-1}		& \qw \\
}$
\end{center}
It is not hard to see that the wires $x_{1}$ and $y_{1}$ in this new comparator circuit will output the same values with the wires $x$ and $y$ respectively in the original circuit. For the general case, we can simply make copies of all wires for each layer of the comparator circuit, where each copy of a wire will be used to carry value of a wire at a single layer of the circuit. Then apply the above construction to simulate the effect of each gates. Note that additional comparator gates are also needed to forward the values of the wires from one layer to another,
in the same way that
we use the gate $\seq{y_{0},y_{1}}$ to forward the value carried in wire $y_{0}$ to wire $y_{1}$ in the above construction. 
\end{proof}

\subsection{The stable marriage problem}
An instance of the stable marriage problem ($\SMP$) is given by a
number $n$ (specifying the number of men and the number of women),
together with a preference list for each man and each woman
specifying a total ordering on all people of the opposite sex.
The goal of $\SMP$ is to produce a perfect matching between men and women,
i.e., a bijection from the set of men to the set of women, such that the
following \emph{stability} condition is satisfied: there are no two
people of the opposite sex who like each other more than their current
partners.  Such a stable solution always exists, but it may not be unique.
Thus $\SMP$ is a search problem (\ref{eq:search}),
rather than a decision problem.

However there is always a unique {\em man-optimal} and a unique
{\em woman-optimal} solution.   In the man-optimal solution each
man is matched with a woman whom he likes at least as well as any
woman that he is matched with in any stable solution.  Dually
for the woman-optimal solution.  Thus we define the
{\em man-optimal stable marriage decision problem} ($\MOSM$) as follows:
given an instance of the stable marriage problem together with
a designated man-woman pair, determine whether that pair is 
married in the man-optimal stable marriage.
We define the {\em woman-optimal stable marriage decision problem}
($\WOSM$) analogously.


We show here that the search version and the decision versions are
computationally equivalent, and each is complete for $\CC$ with
respect to the appropriate reducibility in Definition \ref{d:manyOne}.
\secref{yuli} shows how to reduce the
lexicographical first maximal matching problem (which is complete for
$\CC$) to the $\SMP$ search problem, and Section~\ref{s:SMP-CCV}
shows how to reduce both the $\MOSM$ and $\WOSM$ problems to $\CCV$. 

\subsection{Notation}
We write the notation ``($T \proves$)'' in front of the statement of a theorem to indicate that the  statement is formulated and proved within the theory $T$. 

We often use two-dimensional matrices to encode binary relations, e.g. the
edge relation of a graph, matching, etc. In this paper, it is more convenient for our purpose to index the entries of matrices starting from 0 instead of 1. In other words, if $X_{n\times n}$ is a two-dimensional matrix, then entries of $X$ consist of all $X(i,j)$ for $0\le i,j <n$, and $X(0,0)$ is the top leftmost entry of $X$.

\section{The new theory $\VCC$} \label{s:VCC}
We encode a comparator circuit as a sequence of pairs $\seq{ i,j }$, where each pair $\seq{ i,j }$ encodes a comparator gate that swaps the values of the wires $i$ and $j$ iff the value on wire $i$ is greater than the value of wire $j$. We also allow ``dummy'' gates of the form $\seq{i,i}$, which do nothing. We want to define a formula $\ccv(m,n,X,Y,Z)$, where 
\begin{itemize}
 \item $X$ encodes a comparator circuit with $m$ wires and $n$ gates as sequence of $n$ pairs $\seq{i,j}$ with $i,j<m$, and we write $(X)^{i}$ to denote the $i^{\rm th}$ comparator gate of the circuit encoded by $X$.
 \item $Y(i)$ encodes the input value for the $i^{\rm th}$ wire as a truth value,
and
 \item $Z$  is an $(n+1)\times m$ matrix, where $Z(i,j)$ is the value of wire $j$ at layer $i$, where each layer is simply a sequence of values carried by all the wires right after a comparator gate. 
\end{itemize}
Although $X$ encodes a circuit with only $n$ gates, the matrix $Z$ actually encodes $n+1$ layers since we use the first layer to encode the input values of the comparator circuit. The formula  $\ccv(m,n,X,Y,Z)$ holds iff $Z$ encodes the correct values of the layers computed by the comparator circuit encoded by $X$ with input $Y$, and thus is defined as  the following $\Sigma_0^B$-formula:
\begin{align}  
&	\forall i<m  \bigl(Y(i)\lra  Z(0,i)\bigr) \wedge
  \forall i<n\forall x<m \forall y < m,   \notag\\
&\indent   (X)^{i}=\seq{x,y} \rightarrow
       \left[\begin{array}{ll}
			&Z(i+1,x)\lra \bigl(Z(i,x)\wedge Z(i,y)\bigr)\\
		\wedge	&Z(i+1,y)\lra \bigl(Z(i,x)\vee Z(i,y)\bigr)\\
		\wedge  &\forall j<m
				\bigl[(j\not=x\wedge j\not=y)
			\rightarrow \bigl( Z(i+1,j)\lra Z(i,j)\bigr)\bigr]
		\end{array}\right] \label{eq:ccv}
\end{align}

\begin{definition}The theory $\VCC$ has vocabulary $\mathcal{L}_A^2$ and is axiomatized by the axioms of $\V^0$ and the following axiom
\begin{align}
&\CCVA:&\exists Z\le \seq{m,n+1}+1,\ \ccv(m,n,X,Y,Z), \label{eq:ccva}
\end{align}
where $\ccv(m,n,X,Y,Z)$ is defined as in \eref{ccv}.
\end{definition}

There is a technical lemma required to show that $\VCC$ fits the
framework described in \cite[Chapter~9]{CN10}.
Define $F_{\CCVB}(m,n,X,Y)$ to be the $Z$ satisfying 
$\ccv(m,n.X,Y,Z)$ (with each $Z(i)$ set false when it is not determined).  
We need to show that the {\em aggregate}
$F^\ast_{\CCVB}$ of $F_{\CCVB}$ is $\Sigma^B_1$-definable in $\VCC$,
where (roughly speaking) $F^\ast_{\CCVB}$ is the string function that
gathers the values of $F_{\CCVB}$ for a polynomially long sequence of
arguments.  The nature of $\CC$ circuits makes this easy:  The
sequence of outputs for a sequence of circuits can be obtained
from a single circuit which computes them all in parallel:  the
lines of the composite circuit comprise the union of the lines of
each component circuit.

Thus the framework of \cite[Chapter 9]{CN10} does apply to $\VCC$,
and in particular the theory $\overline{\VCC}$ is a universal
conservative extension of $\VCC$ whose function symbols are precisely
those in the function class $\FCC$.  

It is hard to work with $\VCC$ up to this point since we have not shown whether $\VCC$ can prove the definability of basic counting functions (as in $\VTC^{0}$). However, we have the following theorem.
\begin{theorem}[$\VNC^{1} \subseteq \VCC$] \label{theo:NC2CC}
The theory $\VCC$ proves the axiom $\MFV$  defined in \eref{mfv}.
\end{theorem}
\begin{proof} Observe that each comparator gate can produce simultaneously an AND gate and an OR gate with the only restriction that each of these gates must have fan-out one. However, since all AND and OR gates of a monotone Boolean formula also have fan-out one, each instance of the  Boolean formula value problem is a special case of $\CCV$.
\end{proof}

Since $\VTC^0 \subseteq \VNC^1$ (see (\ref{eq:theories})) the next
result is an immediate consequence of this theorem.

\begin{corollary}[$\VTC^{0}\subseteq \VCC$]\label{cor:VTC}
The theory $\VCC$ proves the axiom $\NUMO$  defined in \eref{numones}. 
\end{corollary}

This corollary is important since it allows us to
use the counting ability of $\VTC^{0}$ freely in $\VCC$ proofs.
In particular using counting and induction on the layers of a comparator
circuit, we can prove in $\VTC^{0}$ the following fundamental property of
comparator circuits.
\begin{corollary}\label{cor:c1}
Given a comparator circuit computation, the theory $\VTC^{0}$ (and hence $\VCC$)  proves that all layers of the computation have the same number of ones and zeros.
\end{corollary}

\begin{theorem}[$\VCC \subseteq \VP$] The theory $\VP$ proves the axiom $\CCVA$  defined in \eref{ccva}.
\end{theorem}
\begin{proof} This easily follows since $\CCV$ is a special case of the \emph{monotone circuit value} problem.
\end{proof}

\section{Lexicographical first maximal matching problem is $\CC$-complete \label{sec:lfcom}}
Let $G=(V,W,E)$ be a bipartite graph, where $V=\set{v_i}_{i=0}^{m-1}$, $W=\set{w_i}_{i=0}^{n-1}$ and $E\subseteq V\times W$. The \emph{lexicographical first maximal matching} (lfm-matching) is the matching produced by successively matching each vertex $v_{0},\ldots,v_{m-1}$ to the least vertex available  in $W$. 

Formally, let $E_{m\times n}$ be a matrix encoding the edge relation of a bipartite graph with $m$ bottom nodes and $n$ top nodes, where $E(i,j)=1$ iff  the bottom node $v_{i}$ is adjacent to the top node $w_{j}$. 
Let $L$ be a matrix of the same size as $E$ with the following intended interpretation: $L(i,j)=1$ iff the edge $(v_i,w_j)$ is in the lfm-matching. We can define a $\Sigma_{0}^{B}$-formula $\lfmm(m,n,X,L)$, which asserts that $L$ properly encodes the lfm-matching of the bipartite graph represented by $X$ as follows:
\begin{align}
\forall i < m \forall j < n, \,L(i,j)\leftrightarrow 
\left[
\begin{array}{ccl}
E(i,j) &\wedge& \forall k<j\, \forall \ell < i \bigl(\neg L(i,k) \wedge \neg L(\ell ,j)\bigr)\\
&\wedge& \forall k<j \bigl(\neg E(i,k) \vee \exists i' < i\, L(i',k)\bigr)
\end{array}\right]. \label{eq:lfmm}
\end{align}

In this section we will show that two decision problems concerning the lexicographical matching of  a bipartite graph are $\CC$-complete (under many-one $\AC^{0}$-reductions).   
The lexicographical  first maximal matching  problem ($\LFMM$) is to decide if a designated 
edge belongs to the lfm-matching of  a bipartite graph $G$.
The vertex version of lexicographical first maximal matching problem ($\VLFMM$) is to decide if a designated top node is matched in the lfm-matching of a bipartite graph $G$. $\LFMM$ is the usual way to define a decision problem for  lexicographical  first maximal matching as seen in \cite{MS92,Sub94}; however, as shown in Sections \ref{sec:c2vl} and \ref{sec:vl2cc}, the $\VLFMM$ problem is even more closely related to the $\CCV$ problem.

We will show that the following two more restricted lexicographical matching problems  are also $\CC$-complete.  We define $\TLFMM$ to be the restriction of $\LFMM$ to bipartite graphs of degree at most three. We define $\TVLFMM$ to be the restriction of $\VLFMM$ to bipartite graphs of degree at most three. 

To show that the problems defined above are equivalent under many-one $\AC^{0}$-reductions, it turns out that we also need the following intermediate problem.  A negation gate flips the value on a wire.  For comparator circuits with negation gates, we allow negation gates to appear on any wire
(see the left diagram of \figref{cn2c} below for an example). 
The comparator circuit value problem with negation gates ($\CCVN$) is: given a comparator circuit with negation gates and input assignment, and a designated wire, decide if that wire outputs 1.

All reductions in this section are summarized in \figref{redu}.

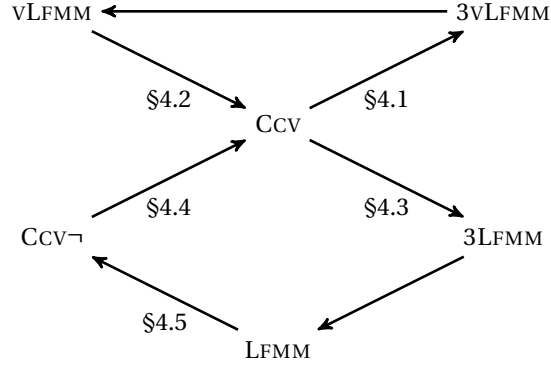
\begin{figure}
\begin{center}
        \tikzstyle{vertex}=
        [%
          	minimum size=5mm,%
          	rectangle,%
	        	thick,%
		text centered
        ]

       \tikzstyle{edge} = [draw,line width=1pt,->]
      \begin{tikzpicture}[>=stealth',scale=1.5]

    	\foreach \pos/\label/\name in 
			{{(0,1)/\VLFMM/vl}, {(2,0)/\CCV/c}, {(4,1)/\TVLFMM/tv},
			{(0,-1)/\CCVN/cvn}, {(2,-2)/\LFMM/el}, {(4,-1)/\TLFMM/te}}
        		\node[vertex] (\name) at \pos {$\label$};
    
    	\foreach \source/ \dest/\label  in {{vl/c/vl2cc},{c/tv/c2vl},{c/te/c2te},{el/cvn/el2cvn},{cvn/c/cvn2c}}
       		\path [edge] (\source) -- node [circle,below,minimum size=5mm,inner sep=1pt] {\S\ref{sec:\label}} (\dest) ;

    	\foreach \source/ \dest  in {{tv/vl},{te/el}}
       		\path[edge] (\source) -- (\dest);

      \end{tikzpicture}
\end{center}

\caption{Label of an arrow denotes the section in which the reduction is described. Arrows without labels denote trivial reductions.}
\label{fig:redu}
\end{figure}

\subsection{$\CCV  \mred^{\AC^{0}}  \TVLFMM$ \label{sec:c2vl}} 
By \propref{p1} it suffices to consider only instances of $\CCV$
in which all comparator gates point upward. We will show that these instances of $\CCV$ are $\AC^{0}$-many-one-reducible to instances of $\TVLFMM$, which consist of bipartite graphs with \emph{degree at most three}.

The key observation is that a comparator gate on the left below closely relates to an instance of $\TVLFMM$ on the right. We use the top nodes $p_{0}$ and $q_{0}$ to represent the values $p_{0}$ and $q_{0}$ carried by the wires $x$ and $y$ respectively before the comparator gate, and the nodes $p_{1}$ and $q_{1}$ to represent the values of $x$ and $y$ after the comparator gate, where a top node is matched iff its respective value is one. 
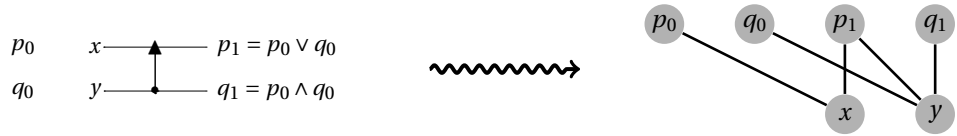
\begin{center}
\begin{minipage}{0.4\textwidth}
\centering
\small 
\hspace{-2cm}$\Qcircuit @C=2.5em @R=1.3em {
\push{\bit{p_{0}}}&\push{x\,} & \ctrlu{1}	& \rstick{p_{1} = p_{0}\vee q_{0}}\qw\\
\push{\bit{q_{0}}}&\push{y\,} & \ctrl{-1} & \rstick{q_{1} = p_{0}\wedge q_{0}}\qw  
}$
\end{minipage}
\begin{minipage}{0.2\textwidth}
\centering
\begin{tikzpicture}
\draw[-to,line width=1.5pt,snake=snake,segment amplitude=.5mm,
         segment length=2mm,line after snake=1mm]  (0,0) -- (2,0);
\end{tikzpicture}
\end{minipage}
\begin{minipage}{0.35\textwidth}
\centering
\tikzstyle{vertex}=[circle,fill=black!30,minimum size=15pt,inner sep=0pt]
\tikzstyle{edge} = [draw,line width=1pt,-]
\begin{tikzpicture}[scale=0.6, auto,swap]
    \foreach \pos/\label/\name in {{(0,2)/p_{0}/p0}, {(2,2)/q_{0}/q0}, {(4,2)/p_{1}/p1},{(6,2)/q_{1}/q1},
                            {(4,0)/x/x}, {(6,0)/y/y}}
        \node[vertex] (\name) at \pos {$\label$};
        
    \foreach \source/ \dest  in {x/p0,x/p1,y/q0,y/p1,y/q1}
        \path[edge] (\source) -- (\dest);

\end{tikzpicture}

\end{minipage}
\end{center}
If nodes $p_{0}$ and $q_{0}$ are not previously matched, i.e. $p_{0}=q_{0}=0$ in the comparator circuit, then edges $\seq{x,p_{0}}$ and $\seq{y,q_{0}}$ are added to the lfm-matching. So the nodes $p_{1}$ and $q_{1}$ are not matched.  If $p_{0}$ is previously matched, but $q_{0}$ is not, then edges $\seq{x,p_{1}}$ and $\seq{y,q_{0}}$ are added to the lfm-matching. So the node $p_{1}$ will be matched but $q_{1}$ will remain  unmatched. The other two cases are similar.

Thus, we can reduce a comparator circuit to the bipartite graph of an $\TVLFMM$ instance  by  converting each  comparator gate into a ``gadget'' described above. We will describe our method through an example, where we are given the comparator circuit in \figref{cex}. 

\begin{wrapfigure}{r}{0.3\textwidth}
\vspace{-.2cm}
\centering
{\small
$\Qcircuit @C=2.2em @R=0.7em {
\push{\bit{0}}&\push{a\,} & \ctrlu{1}	& \qw&    \rstick{\bit{1}} \qw \\
\push{\bit{1}}&\push{b\,} & \ctrl{-1} & \ctrlu{1}&   \rstick{\bit{1}} \qw  \\
\push{\bit{1}}&\push{c\,} & \qw &	\ctrl{-1}&   \rstick{\bit{0}} \qw  \\
		&	    &\lstick{0} & \lstick{1} &\lstick{2}\\
}$}
\vspace{-.15cm}
\caption{}
\label{fig:cex}
\end{wrapfigure}

We divide the comparator circuit into vertical layers 0, 1, and 2 as shown in \figref{cex}. Since the  circuit has three wires $a$, $b$ and $c$, for each layer $i$, we use six nodes, including three top nodes $a_{i}$, $b_{i}$ and $c_{i}$ representing the values of the wires $a$, $b$ and $c$ respectively, and three bottom nodes $a'_{i},b'_{i},c'_{i}$, which are auxiliary nodes used to simulate the effect of the comparator gate at layer $i$.\\
\textbf{Layer 0:} This is the input layer, so we add an edge $\set{x_{i},x'_{i}}$ iff the wire $x$ takes input value 1. In this example, since $b$ and $c$ are wires taking input 1, we need to add the edges $\set{b_{0},b'_{0}}$ and $\set{c_{0},c'_{0}}$. 
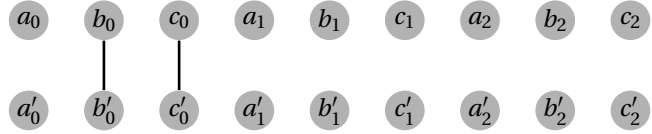
\begin{center}
\tikzstyle{vertex}=[circle,fill=black!30,minimum size=15pt,inner sep=0pt]
\tikzstyle{edge} = [draw,line width=1pt,-]
\begin{tikzpicture}[scale=1, auto,swap]
    \foreach \i in {0,1,2}{
    	\node[vertex] (a\i) at (0+\i*3,1.2) {$a_{\i}$} ; 
	\node[vertex] (b\i) at (1+\i*3,1.2) {$b_{\i}$} ; 
	\node[vertex] (c\i) at (2+\i*3,1.2) {$c_{\i}$} ; 
	
	\node[vertex] (a'\i) at (0+\i*3,0) {$a'_{\i}$} ; 
	\node[vertex] (b'\i) at (1+\i*3,0) {$b'_{\i}$} ; 
	\node[vertex] (c'\i) at (2+\i*3,0) {$c'_{\i}$} ; 
    }
        
    \foreach \source/ \dest  in {b'0/b0,c'0/c0}
        \path[edge] (\source) -- (\dest);

\end{tikzpicture}
\end{center}
\textbf{Layer 1:} We then add the gadget simulating the comparator gate from wire $b$ to wire $a$ as follows.
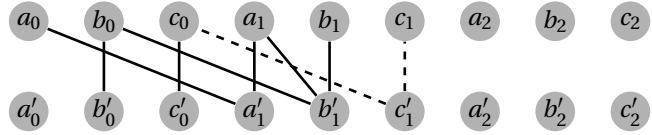
\begin{center}
\tikzstyle{vertex}=[circle,fill=black!30,minimum size=15pt,inner sep=0pt]
\tikzstyle{edge} = [draw,line width=1pt,-]
\tikzstyle{dedge} = [draw,dashed,line width=1pt,-]
\begin{tikzpicture}[scale=1, auto,swap]
    \foreach \i in {0,1,2}{
    	\node[vertex] (a\i) at (0+\i*3,1.2) {$a_{\i}$} ; 
	\node[vertex] (b\i) at (1+\i*3,1.2) {$b_{\i}$} ; 
	\node[vertex] (c\i) at (2+\i*3,1.2) {$c_{\i}$} ; 
	
	\node[vertex] (a'\i) at (0+\i*3,0) {$a'_{\i}$} ; 
	\node[vertex] (b'\i) at (1+\i*3,0) {$b'_{\i}$} ; 
	\node[vertex] (c'\i) at (2+\i*3,0) {$c'_{\i}$} ; 
    }
        
    \foreach \source/ \dest  in {b'0/b0,c'0/c0,a'1/a0,a'1/a1,b'1/b0,b'1/a1,b'1/b1}
        \path[edge] (\source) -- (\dest);

    \foreach \source/ \dest  in {c'1/c0,c'1/c1}
        \path[dedge] (\source) -- (\dest);

\end{tikzpicture}
\end{center}
Since the value of wire $c$ does not change when going from layer 0 to layer 1, we can simply propagate the value of $c_{0}$ to $c_{1}$ using the pair of dashed edges in the picture.\\
\textbf{Layer 2:} We proceed very similarly to layer 1 to get the following bipartite graph.
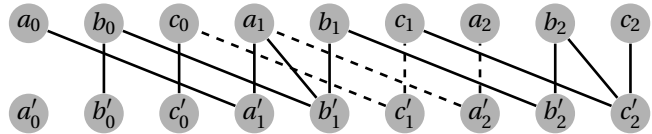
\begin{center}
\tikzstyle{vertex}=[circle,fill=black!30,minimum size=15pt,inner sep=0pt]
\tikzstyle{edge} = [draw,line width=1pt,-]
\tikzstyle{dedge} = [draw,dashed,line width=1pt,-]
\begin{tikzpicture}[scale=1, auto,swap]
    \foreach \i in {0,1,2}{
    	\node[vertex] (a\i) at (0+\i*3,1.2) {$a_{\i}$} ; 
	\node[vertex] (b\i) at (1+\i*3,1.2) {$b_{\i}$} ; 
	\node[vertex] (c\i) at (2+\i*3,1.2) {$c_{\i}$} ; 
	
	\node[vertex] (a'\i) at (0+\i*3,0) {$a'_{\i}$} ; 
	\node[vertex] (b'\i) at (1+\i*3,0) {$b'_{\i}$} ; 
	\node[vertex] (c'\i) at (2+\i*3,0) {$c'_{\i}$} ; 
    }
        
    \foreach \source/ \dest  in 	{b'0/b0,c'0/c0,a'1/a0,
    						a'1/a1,b'1/b0,b'1/a1,b'1/b1,
						b'2/b1,b'2/b2,c'2/c1,c'2/c2,c'2/b2}
        \path[edge] (\source) -- (\dest);

    \foreach \source/ \dest  in {c'1/c0,c'1/c1,a'2/a1,a'2/a2}
        \path[dedge] (\source) -- (\dest);

\end{tikzpicture}
\end{center}
Finally, we can get the output values of the comparator circuit by looking at the ``output'' nodes $a_{2},b_{2},c_{2}$ of this bipartite graph. We can easily check that $a_{2}$ is the only node that remains unmatched, which corresponds exactly to the only zero produced by wire $a$ of the comparator circuit in \figref{cex}.

\begin{remark}
The construction above is an $\AC^0$ many-one reduction since each gate in the comparator circuit can be reduced to exactly one gadget in the bipartite graph that simulates the effect of the comparator gate. Note that  since it can be tedious and unintuitive to work with  $\AC^{0}$-circuits, it might seem hard to justify that our reduction is an $\AC^0$-function. However, thanks to \theoref{szb}, we do not have to work with $\AC^{0}$-circuits directly; instead, it is  not hard to construct a $\Sigma_{0}^{B}$-formula that defines the above reduction. 
\end{remark}

The correctness of our construction can be proved in $\VCC$ by using $\Sigma_{0}^{B}$ induction on the layers of the circuits and arguing that the matching information of the nodes in the bipartite graph can be correctly translated to the values carried by the wires at each layer.

\subsection{$\VLFMM \mred^{\AC^{0}} \CCV$ \label{sec:vl2cc}}
Consider an instance of $\VLFMM$  consisting of a bipartite graph in \figref{l2c}. Recall that we find the lfm-matching by matching the bottom nodes $x, y, \ldots$ successively  to the first available node on the top. Hence we can simulate the matching of the bottom nodes to the top nodes using  comparator circuit on the right of \figref{l2c}, where we can think of the moving of a one, say from wire $x$ to wire $a$, as the matching of node $x$ to node $a$ in the original bipartite graph. In this construction, a top node is matched iff its corresponding wire in the circuit outputs 1.
\begin{figure}[!h]
\centering
\begin{minipage}{0.3\textwidth}
\tikzstyle{vertex}=[circle,fill=black!30,minimum size=15pt,inner sep=0pt]
\tikzstyle{edge} = [draw,line width=1pt,-]
\tikzstyle{cedge} = [draw,line width=4pt,-,black!35]
\begin{tikzpicture}[scale=0.6, auto,swap]
    \foreach \pos/\name in {{(0,2)/a}, {(2,2)/b}, {(4,2)/c},{(6,2)/d},
                            {(0,0)/x}, {(2,0)/y}, {(4,0)/z}}
        \node[vertex] (\name) at \pos {$\name$};
        
    \foreach \source/ \dest  in {x/a,x/b,x/c,y/a,y/c,z/b,z/d}
        \path[edge] (\source) -- (\dest);

    \begin{pgfonlayer}{background}
    \foreach \source/ \dest  in {x/a,y/c,z/b}
        \path[cedge] (\source) -- (\dest);
    \end{pgfonlayer}

\end{tikzpicture}
\end{minipage}
\begin{minipage}{0.18\textwidth}
\centering
\begin{tikzpicture}
\draw[-to,line width=1.5pt,snake=snake,segment amplitude=.5mm,
         segment length=2mm,line after snake=1mm]  (0,0) -- (1.5,0);
\end{tikzpicture}
\end{minipage}
\begin{minipage}{0.45\textwidth}
\centering
{\small$\Qcircuit @C=1.4em @R=0.5em {
\push{\bit{0}}&\push{a\,} & \ctrlu{4}	& \qw 	&\qw		&\qw		&\ctrlu{5}	& \qw 	& \qw	& \qw	& \qw	& \qw	& \qw	& \qw &\rstick{\bit{1}} \qw \\
\push{\bit{0}}&\push{b\,} &	\qw		& \ctrlu{3} &\qw  	&\qw		&\qw		&\qw		&\qw		&\qw		& \qw	& \ctrlu{5}	& \qw	& \qw &\rstick{\bit{1}} \qw\\
\push{\bit{0}}&\push{c\,} &	\qw		& \qw  	&\qw  	&\qw		&\qw 	&\qw		&\ctrlu{3} 	&\qw		&\qw		&\qw		&\qw		&\qw&\rstick{\bit{1}} \qw\\
\push{\bit{0}}&\push{d\,} & \qw	 	& \qw 	& \ctrlu{1} 	&\qw		&\qw		&\qw		&\qw 	&\qw		& \qw	& \qw 	& \qw	& \ctrlu{3}&\rstick{\bit{0}} \qw \\
\push{\bit{1}}&\push{x\,} & \ctrl{-1} 	& \ctrl{-1}  & \ctrl{-1} 	&\ctrl{0}	& \qw	&\qw		&\qw 	&\qw		&\qw		&\qw		& \qw	& \qw&    \rstick{\bit{0}} \qw \\
\push{\bit{1}}&\push{y\,} & \qw 		& \qw 	&\qw		&\qw		& \ctrl{-1}  &\ctrl{0}	&\ctrl{-1}	&\ctrl{0}	& \qw	& \qw 	& \qw	& \qw&\rstick{\bit{0}} \qw \\
\push{\bit{1}}&\push{z\,} &	\qw		& \qw  	&\qw  	&\qw		& \qw	&\qw		&\qw 	&\qw		&\ctrl{0}	&\ctrl{-1}	& \ctrl{0}	&\ctrl{-1}&\rstick{\bit{0}} \qw
\gategroup{1}{3}{7}{6}{1em}{-}
\gategroup{1}{7}{7}{10}{1em}{-}
\gategroup{1}{11}{7}{14}{1em}{-}
}$}
\end{minipage}
\caption{}
\label{fig:l2c}
\end{figure}

Note that we draw bullets without any arrows going out from them in the circuit to  denote dummy gates, which do nothing. These dummy gates are introduced for the following technical reason. Since the bottom nodes might not have the same degree, the position of a comparator gate really depends on the number of edges that do not appear in the bipartite graph, which makes it harder to give a direct $\AC^{0}$-reduction. By using dummy gates, we can treat the graph as if it is a complete bipartite graph, where the missing edges are represented by dummy gates. This can easily be shown to be an $\AC^0$-reduction from $\VLFMM$ to $\CCV$, and its correctness can be carried
out in $\VCC$ using $\Sigma^B_0$-induction on the layers of the circuit.
This together with the reduction from \secref{c2vl} gives us the following theorem.

\begin{theorem}\label{theo:lfcom} 
($\VCC\proves$) The problems $\CCV$, $\TVLFMM$  and  $\VLFMM$ are equivalent under many-one $\AC^0$-reductions.
\end{theorem}

%

\subsection{$\CCV  \mred^{\AC^{0}}  \TLFMM$ \label{sec:c2te}}
We start by applying the reduction  $\CCV  \mred^{\AC^{0}}  \TVLFMM$
of \secref{c2vl} to get an instance of $\TVLFMM$, and notice that 
the degrees of the top ``output'' nodes of the resulting bipartite
graph, e.g. the nodes $a_{2},b_{2},c_{2}$ in the example of \secref{c2vl},
have degree at most two.
Now we show how to reduce such instances of $\TVLFMM$ (i.e. those
whose designated top vertices have degree at most two) to $\TLFMM$.
Consider the graph $G$ with degree at most three and a designated top
vertex $b$ of degree two as shown on the left of \figref{vl2el}. We construct a bipartite graph $G'$, which contains a copy of the graph $G$ and one additional top node $w_t$ and one additional bottom node $w_b$, and two edges $\set{b, w_b}$ and $\set{w_t, w_b}$, as shown in \figref{vl2el}. Observe that the degree of the new graph $G'$ is at most three.
\begin{figure}[!h]
\begin{center}
\begin{minipage}{0.3\textwidth}
\centering 
\tikzstyle{vertex}=[circle,fill=black!30,minimum size=15pt,inner sep=0pt]
\tikzstyle{edge} = [draw,line width=1pt,-]
\tikzstyle{cedge} = [draw,line width=4pt,-,black!35]
\tikzstyle{dedge} = [draw,dashed,line width=1pt,-]
\begin{tikzpicture}[scale=0.6, auto,swap]
    \foreach \pos/\name in {{(0,2)/a}, {(2,2)/b}, {(4,2)/c},
                            {(0,0)/x}, {(2,0)/y}, {(4,0)/z}}
        \node[vertex] (\name) at \pos {$\name$};
        
    \foreach \source/ \dest  in {x/a,x/b,x/c,y/a,y/c,z/b}
        \path[edge] (\source) -- (\dest);

\end{tikzpicture}
\end{minipage}
\begin{minipage}{0.2\textwidth}
\centering
\begin{tikzpicture}
\draw[-to,line width=1.5pt,snake=snake,segment amplitude=.5mm,
         segment length=2mm,line after snake=1mm]  (0,0) -- (2,0);
\end{tikzpicture}
\end{minipage}
\begin{minipage}{0.4\textwidth}
\centering
\tikzstyle{vertex}=[circle,fill=black!30,minimum size=15pt,inner sep=0pt]
\tikzstyle{edge} = [draw,line width=1pt,-]
\tikzstyle{dedge} = [draw,dashed,line width=1pt,-]

\begin{tikzpicture}[scale=0.6, auto,swap]
    \foreach \pos/\name in {{(0,2)/a}, {(2,2)/b}, {(4,2)/c}, {(6,2)/w_t}, 
    					{(0,0)/x}, {(2,0)/y}, {(4,0)/z},{(6,0)/w_b}}
        \node[vertex] (\name) at \pos {$\name$};
        
    \foreach \source/ \dest  in {x/a,x/b,x/c,y/a,y/c,z/b}
        \path[edge] (\source) -- (\dest);
        
    \foreach \source/ \dest  in {{w_b/b},{w_b/w_t}}
        \path[dedge] (\source) -- (\dest);

\end{tikzpicture}
\end{minipage}
\end{center}
\caption{}
\label{fig:vl2el}
\end{figure}

We treat the resulting bipartite graph $G'$ and the edge $\set{w_{t},w_{b}}$ as an instance of $\TLFMM$. It is not hard to see that the vertex $b$ is matched in the lfm-matching of the original bipartite graph $G$ iff the edge $\set{w_t, w_b}$ is in the lfm-matching of the new bipartite graph $G'$.

\subsection{$\CCVN \mred^{\AC^{0}}  \CCV$ \label{sec:cvn2c}}
Recall that a comparator circuit value problem with negation gates ($\CCVN$) is the task of deciding,  on a comparator circuit with negation gates  and an input assignment, if a designated wire outputs one. 
It should be clear that $\CCV$ is a special case of $\CCVN$ and hence $\AC^{0}$ many-one-reducible to $\CCVN$. Here, we show the nontrivial direction that $\CCVN \mred^{\AC^{0}}  \CCV$.

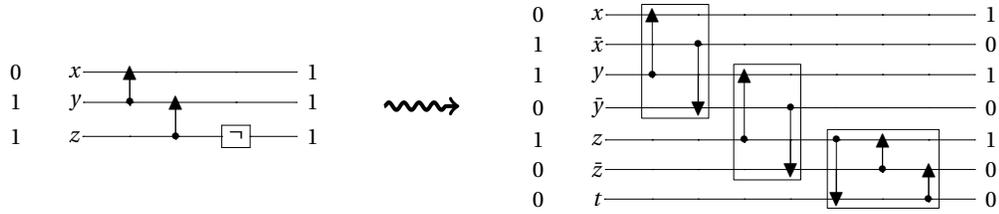
\begin{figure}[!h]
\begin{center}
\begin{minipage}{0.3\textwidth}
\centering
\small 
$\Qcircuit @C=2.2em @R=0.7em {
\push{\bit{0}}&\push{x\,} & \ctrlu{1}	& \qw&  \qw  &\rstick{\bit{1}} \qw \\
\push{\bit{1}}&\push{y\,} & \ctrl{-1} & \ctrlu{1}&  \qw &\rstick{\bit{1}} \qw  \\
\push{\bit{1}}&\push{z\,} & \qw &	\ctrl{-1}&   \gate{\neg}&\rstick{\bit{1}} \qw  
}$
\end{minipage}
\begin{minipage}{0.2\textwidth}
\centering
\begin{tikzpicture}
\draw[-to,line width=1.5pt,snake=snake,segment amplitude=.5mm,
         segment length=2mm,line after snake=1mm]  (0,0) -- (1,0);
\end{tikzpicture}
\end{minipage}
\begin{minipage}{0.4\textwidth}
\centering
\small 
$\Qcircuit @C=2.2em @R=0.7em {
\push{\bit{0}}&\push{x\,}  	& \ctrlu{1}	& \qw	& \qw	& \qw	&  \qw 	& \qw	 & \qw	&\rstick{\bit{1}} \qw \\
\push{\bit{1}}&\push{\bar{x}\,} 	& \qw	& \ctrl{2}	& \qw	& \qw	&  \qw  	& \qw	& \qw	&\rstick{\bit{0}} \qw \\
\push{\bit{1}}&\push{y\,} 	& \ctrl{-1} 	& \qw	& \ctrlu{1}	& \qw	&  \qw 	& \qw	& \qw	&\rstick{\bit{1}} \qw  \\
\push{\bit{0}}&\push{\bar{y}\,} 	& \qw	& \ctrld{-1}	& \qw	& \ctrl{2}	&  \qw  	& \qw	& \qw	&\rstick{\bit{0}} \qw \\
\push{\bit{1}}&\push{z\,} 	& \qw 	& \qw	&\ctrl{-1}	& \qw	&  \ctrl{2}	& \ctrlu{1}	&\qw	&\rstick{\bit{1}} \qw \\
\push{\bit{0}}&\push{\bar{z}\,} 	& \qw	& \qw	& \qw	&\ctrld{-1}	&  \qw  	& \ctrl{-1}	&\ctrlu{1}&\rstick{\bit{0}} \qw \\ 
\push{\bit{0}}&\push{t\,} 	& \qw	& \qw	& \qw	&\qw		& \ctrld{-1}& \qw & \ctrl{-1}	&\rstick{\bit{0}} \qw  
\gategroup{1}{3}{4}{4}{1em}{-}
\gategroup{3}{5}{6}{6}{1em}{-}
\gategroup{5}{7}{7}{9}{1em}{-}
}$
\end{minipage}
\end{center}
\caption{Successive gates on the left circuit corresponds to the successive boxes of gates on the right circuit.}
\label{fig:cn2c}
\end{figure}

This reduction is based on  ``double-rail'' logic.
Given an instance of $\CCVN$ consisting of a comparator circuit with negation gates $C$ with its input $I$ and a designated wire $s$, we construct  an instance of $\CCV$ consisting of a comparator circuit $C'$ with its input $I'$ and a designated wire $s'$ as follows.
For every wire $w$ in $I$ we put in two corresponding wires, $w$ and
$\bar{w}$, in $C'$.  We define input $I'$ of $C'$ such that the input value
of $\bar{w}$ is the negation of the input value of $w$. We want to
fix things so that the value carried by the wire $\bar{w}$ at each layer
is always the negation of the value carried by $w$. For any comparator
gate $\seq{y,x}$ in $C$ we put in both the gate $\seq{y,x}$ and a second
gate $\seq{\bar{x},\bar{y}}$ in $C'$ immediately after $\seq{y,x}$.
It is easy to check by De Morgan's laws that the wires $x$ and $y$ in
$C'$ carry the corresponding values of $x$ and $y$ in $C$, and the
wires $\bar{x}$ and $\bar{y}$ in $C'$ carry the corresponding negations
of $x$ and $y$ in $C$.

The circuit $C'$ has one extra wire $t$ with input value $0$ to help in
translating negation gates.  For each negation gate on a wire, says $z$
in the example from \figref{cn2c}, we add and three comparator gates
$\seq{z,t}$, $\seq{\bar{z},z}$,  $\seq{t,\bar{z}}$ as shown in the right
circuit of \figref{cn2c}. Thus $t$ as a
temporary ``container'' that we use to swap the values of carried by the
wires $z$ and $\bar{z}$. Note that the swapping of values of $z$ and $\bar{z}$ in $C'$ simulates the effect of a negation in $C$.  Also note that
after the swap takes place the value of $t$ is restored to $0$. 

Finally note that the output value of the designated wire $s$ in
$C$ is 1 iff the output value of the
corresponding wire $s$ in $C'$ with input $I'$ is 1.  Thus we set
the designated wire $s'$ in $I'$ to be $s$.

\subsection{$\LFMM \mred^{\AC^{0}} \CCVN$ \label{sec:el2cvn}}
Consider an instance of $\LFMM$  consisting of a bipartite graph on the left of \figref{el2c}, and a designated edge $\set{y,c}$. Without loss of generality, 
we can safely ignore all top vertices occurred after $c$, all bottom vertices occurred after $y$ and all the edges associated with them, since they are not going to affect the outcome of the instance. 
Using the construction from \secref{vl2cc},  we can simulate the matching of the bottom nodes to the top nodes using  comparator circuit in the upper box on the right of \figref{el2c}.
\begin{figure}[!h]
\centering
\begin{minipage}{0.25\textwidth}
\centering
\tikzstyle{vertex}=[circle,fill=black!30,minimum size=15pt,inner sep=0pt]
\tikzstyle{edge} = [draw,line width=1pt,-]
\tikzstyle{cedge} = [draw,line width=4pt,-,black!35]
\begin{tikzpicture}[scale=0.6, auto,swap]
    \foreach \pos/\label/\name in {{(0,2)/a/a}, {(2,2)/b/b}, {(4,2)/c/c},
                            {(0,0)/x/x}, {(2,0)/y/y}}
        \node[vertex] (\name) at \pos {$\label$};
        
    \foreach \source/ \dest  in {x/a,x/b,y/a,y/c}
        \path[edge] (\source) -- (\dest);
        
    \begin{pgfonlayer}{background}
    \foreach \source/ \dest  in {y/c}
        \path[cedge] (\source) -- (\dest);
    \end{pgfonlayer}
    
\end{tikzpicture}
\end{minipage}
\begin{minipage}{0.18\textwidth}
\centering
\begin{tikzpicture}
\draw[-to,line width=1.5pt,snake=snake,segment amplitude=.5mm,
         segment length=2mm,line after snake=1mm]  (0,0) -- (1.5,0);
\end{tikzpicture}
\end{minipage}
\begin{minipage}{0.45\textwidth}
\centering
{\small$\Qcircuit @C=2em @R=0.5em {
\push{\bit{0}}&\push{a\,} & \ctrlu{3}	& \qw 	&\qw		&\ctrlu{4}	&\qw		&\qw			&\qw		&\rstick{\bit{1}} \qw \\
\push{\bit{0}}&\push{b\,} &	\qw		&\ctrlu{2} 	&\qw		&\qw		&\qw 	&\qw 		&\qw		&\rstick{\bit{0}} \qw\\
\push{\bit{0}}&\push{c\,} & \qw 		& \qw 	&\qw		&\qw		&\qw		&\ctrlu{2}		&\ctrl{1}	&\rstick{\bit{1}} \qw \\
\push{\bit{1}}&\push{x\,} & \ctrl{-1} 	& \ctrl{-1}  &\ctrl{0}	&\qw		& \qw         &\qw			& \qw	&\rstick{\bit{0}} \qw \\
\push{\bit{1}}&\push{y\,} & \qw 		& \qw	&\qw 	&\ctrl{-1}	&\ctrl{0}	&\ctrl{-1}		& \qw	&\rstick{\bit{0}} \qw \\
\push{\bit{0}}&\push{a'\,} &\ctrlu{3}	&\qw 	&\qw		&\ctrlu{4}	&\qw  	&\qw			& \qw	&\rstick{\bit{1}} \qw\\
\push{\bit{0}}&\push{b'\,} &\qw		&\ctrlu{2}  	&\qw		&\qw		&\qw 	&\qw			&\qw 	&\rstick{\bit{0}} \qw\\
\push{\bit{0}}&\push{c'\,} & \qw 		& \qw 	&\qw		&\qw		&\qw 	&\gate{\neg}	&\ctrld{-5}	&\rstick{\bit{1}} \qw \\
\push{\bit{1}}&\push{x'\,} & \ctrl{-1} 	& \ctrl{-1}	&\ctrl{0}	&\qw		& \qw	&\qw			& \qw	&\rstick{\bit{0}} \qw \\
\push{\bit{1}}&\push{y'\,} & \qw 		& \qw	&\qw 	&\ctrl{-1}	&\ctrl{0}	&\qw			& \qw	&\rstick{\bit{0}} \qw
\gategroup{1}{3}{5}{8}{1em}{-}
\gategroup{6}{3}{10}{7}{1em}{-}
}$}
\end{minipage}
\caption{}
\label{fig:el2c}
\end{figure}

We keep another running copy of this simulation on the bottom,
(see the wires labelled $a',b',c',x',y'$ in \figref{el2c}). The only
difference  is that the comparator gate $\seq{y',c'}$ corresponding to
the designated edge $\set{y,c}$ is not added. We finally add a negation
gate on $c'$ and  a comparator gate $\seq{c,c'}$.   We let the desired
output of the $\CCV$ instance be the output of $c$, since $c$ outputs 1
iff the edge $\set{y,c}$ is added to the lfm-matching. It  is not hard to see that such construction can be generalized, and the output correctly computes if the designated edge is in the lfm-matching. 

\begin{theorem}\label{theo:lfm}
($\VCC\proves$) The problems $\CCV$,  $\CCVN$, $\TLFMM$  and  $\LFMM$ are equivalent under many-one $\AC^0$-reductions.
\end{theorem}

Combined with the results from Sections \ref{sec:c2vl} and \ref{sec:vl2cc}, we have the following corollary.

\begin{corollary}\label{cor:lfm}
($\VCC\proves$) The problems $\CCV$, $\TVLFMM$, $\VLFMM$, $\CCVN$, $\TLFMM$  and  $\LFMM$ are equivalent under many-one $\AC^0$-reductions.
\end{corollary}

Since $\CCVN$ is complete for $\CC$, we can use comparator circuits to decide the complement of the $\CCV$ problem: given a comparator circuit and and input assignment, does a designated wire output 0? Thus, we have the following corollary.
\begin{corollary} $\CC$ is closed under complementation.
\end{corollary}

\section{The theory $\VCC$ contains $\VNL$}
Each instance of the \textsc{Reachability} problem consists of a directed acyclic  graph $G=(V,E)$, where $V=\set{u_{0},\ldots,u_{n-1}}$, and  we want to decide if there is a path from $u_{0}$ to $u_{n-1}$. It is well-known that \textsc{Reachability} is $\NL$-complete. 
Since a directed graph can be converted into a layered graph with
an equivalent reachability problem, it
suffices to give a comparator circuit construction  that solves instances
of   \textsc{Reachability} satisfying the following assumption: 
\begin{align}
\text{The graph $G$ only has directed edges of the from $(u_{i},u_{j})$, where $i<j$.} \label{eq:a1}
\end{align}
We believe that our new construction for showing that $\NL\subseteq \CC$ is more intuitive than the one in \cite{Sub90,MS92}. Moreover, we reduce \textsc{Reachability} to $\CCV$ directly without going through some intermediate complete problem, and this was stated as an open problem in \cite[Chapter 7.8.1]{Sub90}.

We will demonstrate our construction through a simple example, where we have the directed graph in \figref{g}  satisfying the assumption \eref{a1}. We will build a comparator circuit as in \figref{NL2CC}, where the wires $\nu_0,\ldots,\nu_4$ represent the vertices $u_{0},\ldots,u_{4}$ of the preceding graph and the wires $\iota_0,\ldots,\iota_4$ are used to feed $1$-bits into the wire $\nu_0$, and from there to the other wires $\nu_{i}$ reachable from $\nu_{0}$. We let every wire $\iota_{i}$ take input 1 and every wire $\nu_{i}$ take input 0. \smallskip 

\begin{wrapfigure}{r}{0.3\textwidth}
\vspace{-3mm}
\begin{center}
\tikzstyle{vertex}=[circle,fill=black!30,minimum size=15pt,inner sep=0pt]
\tikzstyle{edge} = [draw,line width=1pt,-to]
\begin{tikzpicture}[scale=1.2, auto,swap]
    \foreach \pos/\name/\label in {{(0,0)/v0/u_{0}}, {(1,0.5)/v1/u_{2}}, {(1,-0.5)/v2/u_{1}},{(2,0)/v3/u_{3}},						{(2,1)/v4/u_{4}}}
        \node[vertex] (\name) at \pos {$\label$};
        
    \foreach \source/ \dest  in {v0/v1,v0/v2,v1/v3,v1/v4}
        \path[edge] (\source) -- (\dest);
\end{tikzpicture}\end{center}
\vspace{-4mm}
\caption{}
\label{fig:g}
\end{wrapfigure}

We next show how to construct the gadgets in the boxes. For a graph with
$n$ vertices ($n=5$ in our example), the $k^{\rm th}$ gadget is
constructed as follows:
\begin{algorithmic}[1]
\STATE Introduce a comparator gate from wire $\iota_{k}$ to wire $\nu_{0}$
\FOR{$i=0,\ldots,n-1$}
\FOR{$j=i+1,\ldots,n-1$}
\STATE Introduce a comparator gate from $\nu_i$ to $\nu_j$ if  $(u_i,u_j)\in E$, or a dummy gate on $\nu_i$ otherwise.
\ENDFOR
\ENDFOR
\end{algorithmic}
Note that the gadgets are identical except for the first comparator gate.

We only use the loop structure to clarify the order the gates are added. The construction can easily be done in $\AC^{0}$ since the position of each gate can be calculated exactly, and thus all gates can be added independently from one another. Note that for a graph with $n$ vertices, we have at most $n$ vertices reachable from a single vertex, and thus we need $n$ gadgets described above. In our example, there are at most 5 wires reachable from wire $\nu_0$, and thus we utilize the gadget 5 times. 
\begin{figure}[!h]
\centering
{\footnotesize
$\Qcircuit @C=1.4em @R=.5em {
\push{\bit{1}}&\push{\iota_0} 
	&\ctrl{4}	& \qw 	& \qw	& \qw	&\qw
	&\qw	& \qw 	& \qw	& \qw	&\qw	
	&\qw	& \qw 	& \qw	& \qw	&\qw	
	&\qw	& \qw 	& \qw	& \qw	&\qw	
	&\qw	& \qw 	& \qw	& \qw	&\qw		
	&\rstick{\bit{0}} \qw \\
\push{\bit{1}}&\push{\iota_1} 
	&	\qw	&\qw  	&\qw  	& \qw	&\qw	
	&\ctrl{3}& \qw 	& \qw	& \qw	&\qw
	&\qw	& \qw 	& \qw	& \qw	&\qw	
	&\qw	& \qw 	& \qw	& \qw	&\qw	
	&\qw	& \qw 	& \qw	& \qw	&\qw		
	&\rstick{\bit{0}} \qw\\
\push{\bit{1}}&\push{\iota_2} 
	&	\qw	&\qw  	&\qw  	& \qw	&\qw	
	&\qw	& \qw 	& \qw	& \qw	&\qw
	&\ctrl{2}& \qw 	& \qw	& \qw	&\qw
	&\qw	& \qw 	& \qw	& \qw	&\qw
	&\qw	& \qw 	& \qw	& \qw	&\qw		
	&\rstick{\bit{0}} \qw\\
\push{\bit{1}}&\push{\iota_3} 	
	& \qw	& \qw 	& \qw	& \qw	&\qw	
	&\qw	& \qw 	& \qw	& \qw	&\qw
	&\qw	& \qw 	& \qw	& \qw	&\qw
	&\ctrl{1}& \qw 	& \qw	& \qw	&\qw
	&\qw	& \qw 	& \qw	& \qw	&\qw	
	&\rstick{\bit{0}} \qw \\
\push{\bit{1}}&\push{\iota_4} 
	& \qw	& \qw 	& \qw	& \qw	&\qw	
	&\qw	& \qw 	& \qw	& \qw	&\qw
	&\qw	& \qw 	& \qw	& \qw	&\qw
	&\qw	& \qw 	& \qw	& \qw	&\qw	
	&\ctrl{1}& \qw 	& \qw	& \qw	&\qw
	&\rstick{\bit{0}} \qw \\	
\push{\bit{0}}&\push{\nu_0} 
	&\ctrld{-1}&\ctrl{0}&\ctrl{0}& \qw	
	&\qw	&\ctrld{-1}&\ctrl{0}&\ctrl{0}& \qw
	&\qw	&\ctrld{-1}&\ctrl{0}&\ctrl{0}& \qw
	&\qw	&\ctrld{-1}&\ctrl{0}&\ctrl{0}& \qw
	&\qw	&\ctrld{-1}&\ctrl{0}&\ctrl{0}& \qw
	&\qw	&\rstick{\bit{1}} \qw \\
\push{\bit{0}}&\push{\nu_1} 
	& \qw 	&\ctrld{-1}& \qw& \qw	&\qw
	& \qw 	&\ctrld{-1}& \qw& \qw	&\qw	
	& \qw 	&\ctrld{-1}& \qw& \qw	&\qw	
	& \qw 	&\ctrld{-1}& \qw& \qw	&\qw	
	& \qw 	&\ctrld{-1}& \qw& \qw	&\qw		
	&\rstick{\bit{1}} \qw \\
\push{\bit{0}}&\push{\nu_2} 
	& \qw 	&\qw 	&\ctrld{-2}&\ctrl{0}&\ctrl{0}
	& \qw 	&\qw 	&\ctrld{-2}&\ctrl{0}&\ctrl{0}
	& \qw 	&\qw 	&\ctrld{-2}&\ctrl{0}&\ctrl{0}
	& \qw 	&\qw 	&\ctrld{-2}&\ctrl{0}&\ctrl{0}
	& \qw 	&\qw 	&\ctrld{-2}&\ctrl{0}&\ctrl{0}
	&\rstick{\bit{1}} \qw \\
\push{\bit{0}}&\push{\nu_3} 
	&	\qw	&\qw  	&\qw  	&\ctrld{-1}&\qw
	&	\qw	&\qw  	&\qw  	&\ctrld{-1}&\qw  
	&	\qw	&\qw  	&\qw  	&\ctrld{-1}&\qw  
	&	\qw	&\qw  	&\qw  	&\ctrld{-1}&\qw  
	&	\qw	&\qw  	&\qw  	&\ctrld{-1}&\qw  	
	&\rstick{\bit{1}} \qw\\
\push{\bit{0}}&\push{\nu_4} 
	& \qw 	&\qw 	&\qw 	&\qw	&\ctrld{-2}
	& \qw 	&\qw 	&\qw 	&\qw	&\ctrld{-2}
	& \qw 	&\qw 	&\qw 	&\qw	&\ctrld{-2}
	& \qw 	&\qw 	&\qw 	&\qw	&\ctrld{-2}	
	& \qw 	&\qw 	&\qw 	&\qw	&\ctrld{-2}
	&\rstick{\bit{1}} \qw 
\gategroup{1}{3}{10}{7}{1em}{-}
\gategroup{1}{8}{10}{12}{1em}{-}
\gategroup{1}{13}{10}{17}{1em}{-}
\gategroup{1}{18}{10}{22}{1em}{-}
\gategroup{1}{23}{10}{27}{1em}{-}
}$}
\caption{A comparator circuit that solves \textsc{Reachability}. (The dummy gates are omitted.)}
\label{fig:NL2CC}
\end{figure}
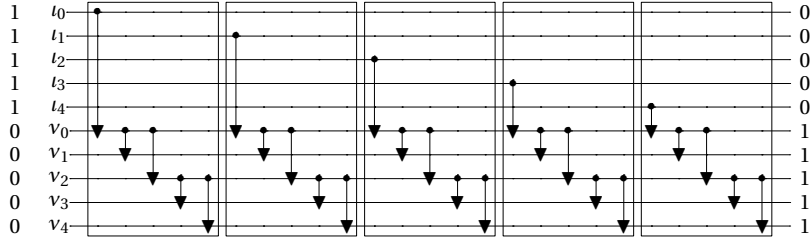

Intuitively, the construction works since each gadget from a box looks for the \emph{lexicographical first maximal path} starting from $v_0$ (with respect to the \emph{natural lexicographical ordering} induced by the vertex ordering $v_0,\ldots,v_n$), and then the vertex at the end of the path will be marked (i.e. its wire will now carry 1) and thus excluded from the search of the gadgets that follow. For example, the gadget from the left-most dashed box in  \figref{NL2CC} will move a value 1 from wire $\iota_0$ to wire $\nu_0$ and from wire $\nu_0$ to wire $\nu_1$. This essentially ``marks'' the wire $\nu_1$ since we cannot move this value 1 away from $\nu_1$, and thus $\nu_{1}$ can no longer receive any new incoming 1. Hence, the gadget from the second box in \figref{NL2CC} will repeat the process of finding the lex-first maximal path  from $v_0$ to the remaining (unmarked) vertices. These searches end when all vertices reachable from $v_0$ are marked. Note that this has the same effect as applying the \emph{depth-first search} algorithm to find all the vertices reachable from $v_{0}$. 
\begin{theorem}[$\VNL \subseteq \VCC$] \label{theo:NL2CCV}
The theory $\VCC$ proves the axiom $\CONN$  defined in \eref{conn}. 
\end{theorem}
\begin{proof}
Recall that if $G$ is a directed graph on $n$ vertices $u_{0},\ldots,u_{n-1}$ with the edge relation $E$, then the formula  $\conn(n,E,U)$ holds iff $U$ is a matrix of $n$ rows, where row $d$ encodes the set of all vertices in $G$ that are reachable from $u_{0}$ using a path of length at most $d$.

We start by converting $G$ into a layered graph $G'=(V',E')$ which
satisfies \eref{a1}, where 
\begin{align}\label{eq:Gprime}
V' &=\bset{u_{i}^{\ell}\mid 0\le i,\ell <n},\\
E' &= \bset{ (u^\ell_i,u^{\ell+1}_j) \mid 0\le i,j,\ell <n
    \mbox{ and ($i=j$ or $(u_i,u_j) \in E$)}}. \nonumber
\end{align}
Observe that a vertex $u_{i}^{\ell}$ is reachable from $u_{0}^{0}$ in $G'$ iff $u_{i}$ is reachable from $u_{0}$ by a path of length at most $\ell$ in $G$. Moreover, if we enumerate the vertices of $G'$ by layers, then $G'$ satisfies the condition \eref{a1}.  We now apply the above construction to $G'$ to find all vertices in $V'$ reachable from $u_{0}^{0}$. Then we construct a matrix $U$ witnessing the formula $\conn(n,E,U)$ by letting the row $d$ encode the set of vertices $u_{i}$ in $V$ such that $u_{i}^{d}$ is reachable from $u_{0}^{0}$ in $G'$. 

We want to show that the comparator circuit constructed for $G'$ using the above method produces the correct set of vertices $u_{i}^{d}$ reachable from $u_{0}^{0}$ for every $d$. Although the correctness of the construction follows from the intuition that the circuit simulates the depth-first search algorithm, we cannot formalize this intuition directly since it would require $\VNL$ reasoning. Recall that up to this point, we only know that $\VTC^{0}\subseteq \VNC^{1}\subseteq \VCC$ by \theoref{NC2CC}. It turns out that using only the counting ability of $\VTC^{0}$, we can analyze the computation of the circuit  in the above construction and argue that since we feed in as many 1-bits as the number of wires representing the vertices of $G'$, these 1-bits will eventually fill all the lower wires that are reachable from $\nu_{0}^{0}$. 

Now we begin the detailed proof.  (To follow the proof it will be
helpful to keep an eye on \figref{NL2CC}.)
Let $C$ be the comparator circuit constructed from the layered graph $G'$
of $G$ using the above construction, where $C$ consists of the wires
$\set{\nu_{i}\mid 0\le i,\ell <n^2}$ representing the vertices of $G'$
and wires $\set{\iota_{0},\ldots,\iota_{n^{2}}}$ are used to feed ones
to the wire $\nu_{0}$.  It is important to note that we will order the
wires from high to low by layers of $G'$ so that the sequence of
wires
\[
   \nu_0,\ldots,\nu_{n^2-1}
\]
corresponds to the sequence of nodes
\[
u^{0}_{0},u^{0}_{1},\ldots,u^{0}_{n-1},u^{1}_{0},u^{1}_{1},\ldots,u^{1}_{n-1},
\ldots, u^{n-1}_{0},u^{n-1}_{1},\ldots,u^{n-1}_{n-1}.
\]
Let $\gamma_0, \gamma_1,\ldots,\gamma_{n^2-1}$ be the successive
gadgets in the circuit $C$.

\begin{lemma}\label{lem:gadget}
($\VCC\proves$) For each $k \le n^2-1$, if wire $\nu_0$ has value $1$ at
the input of gadget $\gamma_k$ then the value of each wire
$\nu_i, 0\le i \le n^2-1$, is the same at the output of $\gamma_k$ as
at the input of $\gamma_k$.  If $\nu_0$ has value $0$ at the input of
$\gamma_k$ then the above is true of $\gamma_k$ with exactly one exception:
some wire $\nu_j$ is $0$ at the input of $\gamma_k$ and $1$ at the output
of $\gamma_k$.
\end{lemma}
\begin{proof}
Note that for $k< n^2-1$ the output values of $\gamma_k$ are the same as
the input values of $\gamma_{k+1}$.

We proceed by induction on $k$.  For $k=0$ the input values of wires 
$(\nu_0,\ldots,\nu_{n^2-1})$ are 0,
but after the first gate $(\iota_0,\nu_0)$ the value
of wire $\nu_0$ becomes 1. Hence by Corollary \ref{cor:c1} applied
to $\gamma_0$ starting after that gate (or by induction on the depth
of the gates in $\gamma_0$), the output values of
wires $(\nu_0,\ldots,\nu_{n^2-1})$ contain a single 1, and the rest are 0.

For the induction step suppose $k>0$.  Suppose first that the value of
$\nu_0$ at the input of $\gamma_k$ is 1, so the output value of $\nu_0$
for $\gamma_{k-1}$ is 1.   Then it follows from the
induction hypothesis that the tuple of values for wires 
$(\nu_0,\ldots,\nu_{n^2-1})$ is either the same for the input and 
output of $\gamma_{k-1}$ or wire $\nu_0$ is the only exception (it
must be an exception if its input value is 0 because by assumption
its output value is 1).  Note that after the first gate $(\iota_0,\nu_0)$
in $\gamma_{k-1}$
the value of wire $\nu_0$ is certainly 1 and so at this point in the
gadget $\gamma_{k-1}$ the tuple of values for wires
$(\nu_0,\ldots,\nu_{n^2-1})$ is the same as at this point in $\gamma_k$.
Since gadgets $\gamma_{k-1}$ and $\gamma_k$ are identical after the
first gate, the outputs for wires $(\nu_0,\ldots,\nu_{n^2-1})$ are
the same for the two gadgets. Thus the lemma follows for this case.

Now suppose that the value of $\nu_0$ at the input of $\gamma_k$ (and
hence the output of $\gamma_{k-1}$) is 0.
Then by the induction hypothesis the value of $\nu_0$ at the input of $\gamma_{k-1}$ is
also 0 and the values of $(\nu_0,\ldots,\nu_{n^2-1})$ at the inputs of
$\gamma_{k-1}$ and $\gamma_k$ are identical except some wire
$\nu_j$ ($j \ne 0$) is 0 for $\gamma_{k-1}$ and 1 for $\gamma_k$.
After the first gate in each gadget the value of $\nu_0$ is 1,
and since the gadgets are
identical except for the first gate, it follows by induction
on $p$ that the value of each wire at position $p$ in $\gamma_{k-1}$
is the same as the value of that wire at position $p$ in $\gamma_k$,
except for each $p$ one wire is 0 in $\gamma_{k-1}$ and 1 in
$\gamma_k$.  The lemma follows when $p$ is the output position.
\end{proof}

The next result follows easily from the preceding  lemma by
induction on $k$. 

\begin{lemma}\label{lem:onePersists}
($\VCC\proves$)
If a wire $\nu_i$ has value $1$ at the output of some gadget $\gamma_k$,
then it has value $1$ at the outputs of all succeeding gadgets.
If the tuple of output values for $(\nu_0,\ldots,\nu_{n^2-1})$
is the same for two successive gadgets $\gamma_k$ and $\gamma_{k+1}$
then $\nu_0 =1$ and the tuple remains the same for the outputs of all
succeeding gadgets.
\end{lemma}

The next result is the only place in the proof of Theorem
\ref{theo:NL2CCV} that makes essential use of
the ability of $\VCC$ to count, as in Corollary \ref{cor:c1}.

\begin{lemma}\label{lem:teminate}
($\VCC\proves$)
Wire $\nu_0$ has value $1$ at the output of the circuit $C$.
\end{lemma}
\begin{proof}
By Lemma \ref{lem:onePersists}, if $\nu_0$ has value 0 at the output
of $C$, then it has value 0 at the output of every gadget $\gamma_k$.
Since there are $n^2$ ones at the input to $C$ (one for every wire
$\iota_i$), by Corollary \ref{cor:c1} there are $n^2$ ones at the
output of $C$.  But if $\nu_0$ remains 0 after every gadget, then
by the construction of $C$ we see that the final output of every
wire $\iota_i$ is 0, and hence outputs of all $n^2$ wires $\nu_i$
must be 1, including $\nu_0$.
\end{proof}

\begin{lemma}\label{lem:noChange}
($\VCC\proves$)
In the final gadget, no wire $\nu_i$ changes its value after any
comparator gate, except possibly $\nu_0$ changes from $0$ to $1$
after the initial gate feeds a $1$ from $\iota_{n^2-1}$.
\end{lemma}

\begin{proof}
We use induction on $i$.  For $i=0$, by Lemma \ref{lem:teminate}
the output of $\nu_0$ in the final gadget is 1.  At the input to this
gadget $\nu_0$ can be either 0 or 1, but after the first gate it is
certainly 1.   The value of $\nu_0$ cannot change from 1 to 0 in
the gadget, because after the first gate there is no further
gate leading down to $\nu_0$ which could bring down a 1 to change the value of 
$\nu_0$ from 0 to 1, but we know the final output value is 1.

For $i>0$ it follows from Lemmas \ref{lem:gadget} and \ref{lem:teminate}
that the value of wire $\nu_i$ is the same at the input and output
of the final gadget.   We note that all gates leading down
to $\nu_i$ preceed all gates leading away from $\nu_i$.  The only
way that $\nu_i$ can change from 0 to 1 is at a gate bringing
down a 1 from above, but that would change the value of the wire
above, violating the induction hypothesis.  The only way that
$\nu_i$ can change from 1 to 0 is at a gate leading away from $\nu_0$,
but then $\nu_i$ cannot change back to 1, so the input and output
cannot be the same.
\end{proof}

\begin{lemma}\label{lem:mon}
($\VCC\proves$)
If a wire $\nu_i$ has value $1$ at some position $p$ in some gadget
$\gamma_k$, then the output of $\nu_i$ in the final gadget is $1$.
\end{lemma}

\begin{proof}
By Lemma \ref{lem:gadget} the tuple of input values for wires
$\nu_0,\ldots,\nu_{n^2-1}$ is the same for gadgets $\gamma_k$ and
$\gamma_{k+1}$ except possibly some input changes from 0 to 1.
From this and the fact that the gadgets $\gamma_k$ and $\gamma_{k+1}$
are the same except for the first gate, it is easy to prove by induction
on $p$ that at position in $p$ each wire has the same value in
$\gamma_{k+1}$ as in $\gamma_k$, except the value might be 0 in
$\gamma_k$ and 1 in $\gamma_{k+1}$.

Therefore by induction on $k$, if some wire $\nu_i$ has value 1 in
position $p$ in some gadget then wire $\nu_i$ has value 1 in position
$p$ in the final gadget.  In this case, by Lemma \ref{lem:noChange}
the output of $\nu_i$ in the final gadget is 1.
\end{proof}

\begin{lemma}\label{lem:corresp}
($\VCC\proves$)
Let $j>0$ and let $y$ be the node in $G'$ corresponding to wire $\nu_j$.
Then the output of $\nu_j$ in $C$ is $1$ iff there is $i<j$ such that
the output of wire $\nu_i$ in $C$ is $1$ and there is an edge from
$x$ to $y$, where $x$ is the node in $G'$ which corresponds to $\nu_i$.
\end{lemma}
\begin{proof}
For the direction ($\Leftarrow$), suppose that the output of wire
$\nu_i$ is 1 and there is an edge from $x$ to $y$ in $G'$.  Then
there is a gate from $\nu_i$ to $\nu_j$ in the final gadget.
Then the value of $\nu_j$ must be 1 by Lemma \ref{lem:noChange},
since otherwise this gate would change $\nu_i$ from 1 to 0.
Since the value of $\nu_j$ cannot change, its value is 1 at the output.

For the direction ($\Rightarrow$) suppose the value of $\nu_j$ is 1
at the output.  Since the value of $\nu_j$ at the input to $C$ is 0,
there is some gadget $\gamma_k$ such that the value of $\nu_j$
changes from 0 to 1.  For this to happen there must be some 
comparator gate in $\gamma_k$ from some wire $\nu_i$ down to $\nu_j$
with $i<j$ such that the value of $\nu_i$ before the gate is 1,
and there is an edge from the node $x$ corresponding to $\nu_i$
to node $y$.  By Lemma \ref{lem:mon} the value of $\nu_i$ at the
output of the final gadget is 1.
\end{proof}

To complete the proof of Theorem \ref{theo:NL2CCV} recall the
meaning of the array $U$ given in (\ref{eq:connMean}), and the
relation between the graph $G = (V,E)$ and
the graph $G'=(V',E')$ given by (\ref{eq:Gprime}).  We define
$U(d,i)$ (the truth value of node $u_i$ of $G$ in row $d$ of the array
$U$) to be true iff the the output of wire $\nu_j$ in circuit $C$ is 1,
where $\nu_j$ is the wire corresponding to node $u^d_i$ in $G'$. 

We prove  that this definition of $U$ satisfies the formula $\conn(n,E,U)$
(\ref{eq:connMean}) (which appears in the axiom (\ref{eq:conn}) for $\VNL$)
by induction on $d$. 
The base case is
$d=0$.  We have $U(0,0)$ holds because the output of $\nu_0$ (the wire
corresponding to node $u^0_0$) is 1 by 
Lemma \ref{lem:teminate}.  For $i>0$, $U(0,i)$ is false because
there is no edge in $G'$ leading to node $u^0_i$, and hence there is
no comparator gate in any gadget in $C$ leading down to the wire
corresponding to $u^0_i$, so that wire has output 0.

The induction step follows directly from our definition of $U$ above,
together with (\ref{eq:connMean}) and Lemma \ref{lem:corresp}.
\end{proof}

As a of  consequence of \theoref{NL2CCV}, we have the following result.

\begin{theorem}\label{theo:redu}
$\CCstar$ is closed under many-one $\NL$-reductions, and hence
$\CCSubr \subseteq \CCstar$.
\end{theorem}
\begin{proof}
This follows from the following three facts:  The function class $\FCCstar$
is closed under composition, $\FNL\subseteq \FCC$, and a decision
problem is in $\CCstar$ iff its characteristic function is in $\FCCstar$.
\end{proof}

\section{The $\SMP$ problem is $\CC$-complete}

\subsection{$\TLFMM$ is $\AC^{0}$-many-one-reducible to  $\SMP$, $\MOSM$ and $\WOSM$ \label{sec:yuli}}

We start by showing that $\TLFMM$ is $\AC^0$-many-one-reducible to
$\SMP$ in the second sense of Definition \ref{d:manyOne}; i.e. we
regard both $\TLFMM$ and $\SMP$ as search problems.  (Of course the
lfm-matching is the unique solution to $\TLFMM$ formulated as a
search problem,, but it is still a total search problem.)

Let $G=(V,W,E)$ be a bipartite graph  from an instance of  $\TLFMM$, where $V$ is the set of bottom nodes, $W$ is the set of top nodes, and $E$ is the edge relation such that the degree of each node is at most three (see the example in the figure on the left  below).  Without loss of generality, we can assume that $|V|=|W|=n$. To reduce it to an instance of $\SMP$, we  double the number of nodes in each partition, where the new nodes are enumerated after the original nodes and the original nodes are enumerated using the ordering of the original bipartite graph, as shown in the diagram on the right below. We also let the bottom nodes and top nodes represent the men and women respectively.
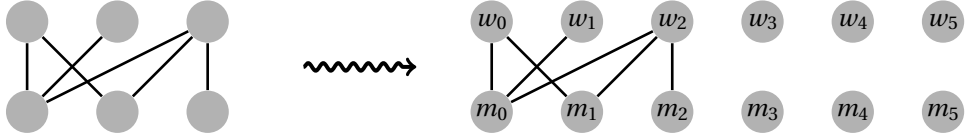
\begin{center}
\tikzstyle{vertex}=[circle,fill=black!30,minimum size=16pt,inner sep=0pt]
\tikzstyle{edge} = [draw,line width=1pt,-]
\begin{minipage}{0.25\textwidth}
\centering
\begin{tikzpicture}[scale=0.6, auto,swap]
    \foreach \pos/\name in {{(0,2)/a}, {(2,2)/b}, {(4,2)/c},
    					{(0,0)/x}, {(2,0)/y}, {(4,0)/z}}
        \node[vertex] (\name) at \pos{};
        
    \foreach \source/ \dest  in {x/a,x/b,x/c,y/a,y/c,z/c}
        \path[edge] (\source) -- (\dest);

\end{tikzpicture}
\end{minipage}
\begin{minipage}{0.2\textwidth}
\centering
\begin{tikzpicture}
\draw[-to,line width=1.5pt,snake=snake,segment amplitude=.5mm,
         segment length=2mm,line after snake=1mm]  (0,0) -- (1.5,0);
\end{tikzpicture}
\end{minipage}
\begin{minipage}{0.5\textwidth}
\begin{tikzpicture}[scale=0.6, auto,swap]
    \foreach \i  in {0,...,5}{
    		\node[vertex](w\i) at (2*\i,2){$w_{\i}$};
		\node[vertex](m\i) at (2*\i,0){$m_{\i}$};
  	}
    \foreach \source/ \dest  in {m0/w0,m0/w1,m0/w2,m1/w0,m1/w2,m2/w2}
        \path[edge] (\source) -- (\dest);

\end{tikzpicture}
\end{minipage}
\end{center}

It remains to define a preference list for each person in this $\SMP$ instance. The preference list of each man $m_{i}$, who represents a bottom node in the original graph, starts with all the woman $w_{j}$ (at most three of them) adjacent to $m_{i}$ in the order that these women are enumerated, followed by all the women $w_{n},\ldots,w_{2n-1}$; the list ends with all women $w_{j}$ not adjacent to  $m_{i}$ also in the order that they are enumerated. For example, the preference list of $m_{2}$ in our example is  $w_{2},w_{3},w_{4},w_{5},w_{0},w_{1}$. The preference list of each newly introduced man $m_{n+i}$ simply consists of $w_{0},\ldots,w_{n-1},w_{n},\ldots,w_{2n-1}$, i.e., in the order that the top nodes are listed. Preference lists for the women are defined dually. 


Intuitively, the preference lists are constructed so that any stable marriage (not necessarily man-optimal) of the new $\SMP$ instance must contain  the lfm-matching of $G$. Furthermore, if a bottom node $u$ from the original graph is not matched to any top node in the lfm-matching of $G$, then the man $m_{i}$ representing $u$ will marry some top node $w_{n+j}$, which is a dummy node that does not correspond to any node of $G$.

Formally, let $\calf{I}$ be an instance of $\SMP$ constructed from a bipartite graph $G=(V,W,E)$ using the above construction, where the set of men is $\set{m_{i}}_{i=0}^{2n-1}$ and the set of women is  $\set{w_{i}}_{i=0}^{2n-1}$, and the preference lists are defined as above. For convenience, assume that the set of bottom nodes and top nodes of $G$ are $V=\set{m_{i}}_{i=0}^{n-1}$ and $W=\set{w_{i}}_{i=0}^{n-1}$ respectively; the set of newly added bottom nodes and top nodes are $V'=\set{m_{i}}_{i=n}^{2n-1}$ and $W'=\set{w_{i}}_{i=n}^{2n-1}$ respectively. We will encode the edge relation of $G$ by a Boolean matrix $E_{n\times n}$, where $E(i,j)=1$ iff $m_i$ is adjacent to $w_j$ in $G$. Similarly, we encode the lfm-matching of $G$ by Boolean matrix $L_{n \times n}$. We encode a stable marriage  by a Boolean matrix $M_{2n\times 2n}$, and thus $M(i,j)=1$ iff $m_i$ marries $w_j$ in $M$. We first prove the following lemma.

\begin{lemma}($\VCC\proves$) 
\label{lem:me}
Given a stable marriage $M$, if $M(i,j)=1$ for some $i,j<n$, then $E(i,j)=1$. 
\end{lemma}
\begin{proof}  
We prove by contradiction. Suppose $M(i,j)=1$ for some $i,j<n$, but $E(i,j)=0$, then since
$M$ is a perfect matching, by the pigeonhole principle $\PHP(n-1,M)$, which is provable in $\VTC^{0}$ (and hence in $\VCC$), we cannot map the set of $n$ men $V'$ into the set of $n-1$ women $W$. Thus, there must exist some $p\ge n$ and $q\ge n$ such that $M(p,q)=1$. Since $m_i$ prefers $w_q$ to $w_j$ and $w_j$ prefers $m_q$ to $m_i$, $M$ is not stable; hence a contradiction. 
\end{proof}

\begin{lemma}\label{lem:me2}
($\VCC\proves$) Let $M$ be a stable marriage of the $\SMP$ instance $\calf{I}$ reduced from the graph $G$, and let $L$ be the lfm-matching of $G$. Then we have $L=\bigl(\set{0,\ldots,n-1}\times \set{0,\ldots,n-1}\bigr)\cap M$, where here $L$ and $M$ are treated as relations.
\end{lemma}

\begin{proof}
First, we will show that $L$ is contained in $\bigl(\set{0,\ldots,n-1}\times \set{0,\ldots,n-1}\bigr)\cap M$. Suppose otherwise, then there must exist a pair $(i,j)$, called ``bad pair", $i,j<n$, such that $L(i,j)=1$ but $M(i,k)=1$ for some $k\ne j$.
Using the  $\Sigma_{0}^{B}\textit{-MIN}$ principle, we pick  the ``bad pair" $(i,j)$ with minimum man index.
There are two cases:
\begin{enumerate}
\item If $k<j$, then we cannot have $L(h,k)=1$ for any $h< i$ for otherwise the pair $(h,k)$ is a ``bad pair" with a smaller man index than $(i,j)$, which is a contradiction. Therefore, $L(h,k)=0$ for any $h<i$. By \lemref{me}, we have $E(i,k)=1$.
Note that for any $\ell< k$, we have $L(i,\ell)=0$, and furthermore, 
by the property of lfm-matching, for any $\ell< k$, either $E(i,\ell)=0$ or there exists some $i'<i$ such that $L(i',\ell)=1$.
Therefore, $L(i,k)=1$ by Eq.~\eref{lfmm}, which contradicts the fact that $L(i,j)=1$.  
\item Otherwise $k>j$, then we cannot have $M(h,j)=1$ for any $h>i$ for otherwise $m_i$ prefers $w_j$ to $w_k$ and $w_j$ prefers $m_i$ to $m_h$, which implies that $M$ is not stable. Therefore, $M(h,j)=1$ for some $h<i$.
By \lemref{me}, we have $E(h,j)=1$.
Note that for any $\ell<j$, $L(h,\ell)=0$, for otherwise the pair $(h,\ell)$ is a ``bad pair" with a smaller man index than $(i,j)$, which is a contradiction. Furthermore, by the property of lfm-matching, for any $\ell< j$, either $E(h,\ell)=0$ or there exists some $i'<i$ such that $L(i',\ell)=1$. 
Since for any $p< h$, we have $L(p,j)=0$, by Eq.~\eref{lfmm}, we have $L(h,j)=1$, which contradicts the fact that $L(i,j)=1$.   
\end{enumerate}

Next, it remains to show that $L$ cannot be strictly contained in $\bigl(\set{0,\ldots,n-1}\times \set{0,\ldots,n-1}\bigr)\cap M$. Suppose otherwise, let $(i,j)$, $i,j<n$, be a pair such that $M(i,j)=1$ but for all $k<n$ and for all $\ell<n$, we have $L(i,k)=0$ and $L(\ell,j)=0$. By \lemref{me}, we have $E(i,j)=1$. 
Furthermore, by the property of lfm-matching, for any $\ell< j$, either $E(i,\ell)=0$ or there exists some $i'<i$ such that $L(i',\ell)=1$. By Eq.~\eref{lfmm}, we have $L(i,j)=1$, which is a contradiction.
\end{proof}

Since \lemref{me2} directly implies the correctness of the many-one $\AC^{0}$-reduction from $\TLFMM$ to $\SMP$ as search problems, any solution of a stable marriage instance constructed by the above reduction provides us all the information to decide if an edge is in the lfm-matching of the original $\TLFMM$ instance. The key explanation is that every instance of stable marriage produced by the above reduction has a unique solution; thus the man-optimal solution coincides with the woman-optimal solution.   Further
\lemref{me2} also shows that the decision version of $\TLFMM$ is 
$\AC^{0}$-many-one-reducible to either of the decision problems
$\MOSM$ and $\WOSM$. Hence we have proven the following theorem.

\begin{theorem}\label{theo:me}
($\VCC\proves$) $\TLFMM$  is $\AC^{0}$-many-one-reducible to $\SMP$, $\MOSM$ and $\WOSM$.
\end{theorem}

\subsection{$\MOSM$ and $\WOSM$ are  $\AC^{0}$-many-one-reducible to $\CCV$ \label{s:SMP-CCV}}
In this section, we formalize a reduction from $\SMP$ to $\CCV$ due to Subramanian \cite{Sub90,Sub94}. Subramanian did not reduce  $\SMP$ to $\CCV$ directly, but to the \emph{network stability problem} built from the less standard X gate, which takes two inputs $p$ and $q$ and produces two outputs $p'=p\wedge \neg q$ and $q'=\neg p\wedge q$.  It is important to note that the ``\emph{network}'' notion in Subramanian's work denotes a generalization of circuits  by allowing a connection from the output of a gate to the input of any gate including itself, and thus a network in his definition might contain cycles. An X-network is a network consisting only of X gates under the important restriction that each X gate has fan-out exactly one for each output it computes. The network stability problem for X gate ($\XNS$) is then to decide if an X-network has a stable configuration, i.e., a way to assign Boolean values to the wires of the network so that  the values are compatible with all the X gates of the network. Subramanian showed in his dissertation \cite{Sub90}  that $\SMP$, $\XNS$ and  $\CCV$ are all equivalent under $\log$ space reductions. 

We do not work with $\XNS$ in this paper since networks are less intuitive and do not have a nice graphical representation as do comparator circuits. By utilizing Subramanian's idea, we give a direct $\AC^{0}$-reduction from $\SMP$ to $\CCV$. For this goal, it turns out to be conceptually simpler to go through a new variant of $\CCV$, where the comparator gates are three-valued instead of Boolean.

\subsubsection{$\TCV$ is $\CC$-complete \label{sec:ccvpi}}
We define the $\TCV$ problem similarly to $\CCV$, i.e., we want to decide, on a given input assignment, if a designated wire of a comparator circuit outputs one.  The only difference is that each wire can now take either value 0, 1 or $\ast$, where a wire takes value $\ast$ when its value is not known to be 0 or 1.  The output values of the comparator gate on two input values $p$ and $q$ will be defined as follows.
\begin{center}
$p\wedge q = \begin{cases}
0	& \text{if $p=0$ or $q=0$}\\
1	& \text{if $p=q=1$}\\
\ast	& \text{otherwise.} 
\end{cases}$
\hspace{2cm} 
$p\vee q = \begin{cases}
0	& \text{if $p=q=0$}\\
1	& \text{if $p=1$ or $q=1$}\\
\ast	& \text{otherwise.} 
\end{cases}$
\end{center}
Every instance of $\CCV$ is also an instance of $\TCV$. We will show that every instance of $\TCV$ is $\AC^{0}$-many-one-reducible to an instance of $\CCV$ by using a pair of Boolean wires to represent each three-valued wire and adding comparator gates appropriately to simulate three-valued comparator gates. 

\begin{theorem}\label{theo:tcv}
($\VCC\proves$) $\TCV$ and $\CCV$ are equivalent under many-one $\AC^0$-reductions.
\end{theorem}

\begin{proof}Since each instance of $\CCV$ is a special case of $\TCV$, it only remains to show that every instance of $\TCV$ is $\AC^{0}$-many-one-reducible to an instance of $\CCV$. 

First, we will describe a gadget built from standard comparator gates that simulates a three-valued comparator gate  as follows. Each wire of an instance of $\TCV$  will be represented by a pair of wires in an instance of $\CCV$. Each three-valued comparator gate on the left below, where $p,q,p\wedge q,p\vee q \in \set{0,1,\ast}$, can be simulated by  a gadget with two standard comparator gates on the right below.
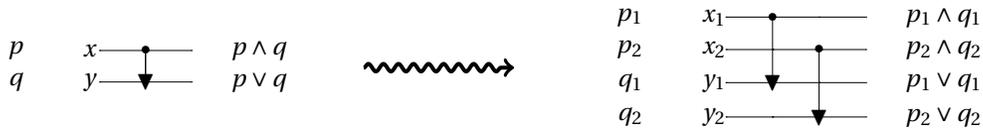
\begin{center}
\begin{minipage}{0.3\textwidth}
\centering
\hspace{-1.5cm}$\Qcircuit @C=2em @R=0.7em {
\lstick{\bit{p}}&\push{x\,} & \ctrl{1}	& \rstick{~~~~~\bit{p\wedge q}}\qw\\
\lstick{\bit{q}}&\push{y\,} & \ctrld{-1} & \rstick{~~~~~\bit{p\vee q}}\qw  
}$
\end{minipage}
\begin{minipage}{0.2\textwidth}
\centering
\begin{tikzpicture}
\draw[-to,line width=1.5pt,snake=snake,segment amplitude=.5mm,
         segment length=2mm,line after snake=1mm]  (0,0) -- (2,0);
\end{tikzpicture}
\end{minipage}
\begin{minipage}{0.4\textwidth}
\centering
$\Qcircuit @C=2em @R=.7em {
\lstick{\bit{p_1}}&\push{x_1} & \ctrl{1}	&\qw	& \rstick{~~~~~\bit{p_1\wedge q_1}}\qw\\
\lstick{\bit{p_2}}&\push{x_2} &\qw	& \ctrl{1}	& \rstick{~~~~~\bit{p_2\wedge q_2}}\qw\\
\lstick{\bit{q_1}}&\push{y_1} & \ctrld{-1} &\qw	& \rstick{~~~~~\bit{p_1\vee q_1}}\qw\\  
\lstick{\bit{q_2}}&\push{y_2} &\qw	& \ctrld{-1} & \rstick{~~~~~\bit{p_2\vee q_2}}\qw
}$
\end{minipage}
\end{center}
The wires $x$ and $y$ are represented using the two pairs of wires $\seq{x_{1},x_{2}}$ and $\seq{y_{1},y_{2}}$, and three possible values 0, 1 and $\ast$ will be encoded by $\seq{0,0}$, $\seq{1,1}$, and $\seq{0,1}$ respectively.  The fact that our gadget  correctly simulates the three-valued comparator gate is shown in the following table.
\begin{center}
\begin{tabular}{|c|c||c|c||c|c||c|c|}
\hline
$~p~$ & $~q~$ & $\seq{p_1,p_2}$ &  $\seq{q_1,q_2}$ & $p\wedge q$ &$p\vee q$ 
& $\seq{p_1\wedge q_1,p_2 \wedge q_2}$& $\seq{p_1\vee q_1,p_2 \vee q_2}$\\
\hline
\hline
0 & 0 & $\seq{0,0}$ &  $\seq{0,0}$ & 0 & 0 & $\seq{0,0}$& $\seq{0,0}$\\
\hline
0 & 1 & $\seq{0,0}$ &  $\seq{1,1}$ & 0 & 1 & $\seq{0,0}$& $\seq{1,1}$\\
\hline
0 & $\ast$ & $\seq{0,0}$ &  $\seq{0,1}$ & 0 & $\ast$ & $\seq{0,0}$& $\seq{0,1}$\\
\hline
1 & 0 & $\seq{1,1}$ &  $\seq{0,0}$ & 0 & 1 & $\seq{0,0}$& $\seq{1,1}$\\
\hline
1 & 1 & $\seq{1,1}$ &  $\seq{1,1}$ & 1 & 1 & $\seq{1,1}$& $\seq{1,1}$\\
\hline
1 & $\ast$ & $\seq{1,1}$ &  $\seq{0,1}$ & $\ast$ & 1 & $\seq{0,1}$& $\seq{1,1}$\\
\hline
$\ast$ & 0 & $\seq{0,1}$ &  $\seq{0,0}$ & 0 & $\ast$ & $\seq{0,0}$& $\seq{0,1}$\\
\hline
$\ast$ & 1 & $\seq{0,1}$ &  $\seq{1,1}$ & $\ast$ & 1 & $\seq{0,1}$& $\seq{1,1}$\\
\hline
$\ast$ & $\ast$ & $\seq{0,1}$ &  $\seq{0,1}$ & $\ast$ & $\ast$ & $\seq{0,1}$& $\seq{0,1}$\\
\hline
\end{tabular} 
\end{center}
Using this gadget, we can reduce an instance of $\TCV$ to an instance of $\CCV$ by doubling the number of wires, and for every three-valued comparator gate of the $\TCV$ instance, we will add a gadget with two standard comparator gates simulating it.  

The above construction shows how to reduce the question
of whether a designated wire ouputs 1 for a given instance of $\TCV$
to the question of whether a \emph{pair} of wires of an instance of $\CCV$
outputs $\seq{1,1}$. However for an instance of $\CCV$ we are only
allowed to decide whether a \emph{single} designated wire outputs 1.
This technical difficulty can be easily overcome since  we can use an $\wedge$-gate (one of the two outputs of a comparator gate) to test whether
a pair of wires outputs $\seq{1,1}$, and outputs the result on a single designated wire.
\end{proof}

\subsubsection{A fixed-point method for solving stable marriage problems}
We formalize a method for solving $\SMP$  using three-valued comparator circuits based on
\cite{Sub90,Sub94}. Consider an instance $\calf{I}$ of $\SMP$
consisting of $n$ men and $n$ women and preference lists for each
man and woman.  From this instance we construct a three-valued comparator
circuit $C_\calf{I}$.  \figref{f2} illustrates $C_\calf{I}$ when
$\calf{I}$ consists of two men $a,b$ and two women $x,y$ with preference
lists given by the matrices.

\begin{figure}[h!]
\centering
\begin{minipage}{.4\textwidth}
\text{Men:}\hspace{1.4cm}
$\begin{array}{c|cc}
a & x & y \\ 
\hline
b & y & x 
\end{array}$\\
\vspace{1cm}\\
\text{Women:}\hspace{1cm}
$\begin{array}{c|cc}
x & a & b \\ 
\hline
y & a & b 
\end{array}$ 
\end{minipage}\hspace{1cm}
\begin{minipage}{0.4\textwidth}
\newcommand{\Z}{0}
\small
$\Qcircuit @C=1.5em @R=0.5em {
\push{\bit{1}}&\push{a_{0}^{i}} &\qw& \ctrl{1}	&\ctrl{9}&\qw 	&\qw 	&\qw 	& \qw\\
\push{\bit{0}}&\push{x_{0}^{i}} &\qw& \ctrld{-1}	&\qw &\qw 	&\ctrl{11}&\qw 	&\qw\\  
\push{\bit{\ast}}&\push{a_{1}^{i}}&\qw & \ctrl{1}	&\qw	&\qw &\qw 	&\qw 	& \qw\\
\push{\bit{0}}&\push{y_{0}^{i}} &\qw& \ctrld{-1} 	&\qw	&\qw &\qw 	&\ctrl{11}&\qw\\  
\push{\bit{\ast}}&\push{b_{1}^{i}}&\qw & \ctrl{1}	&\qw	&\qw &\qw 	&\qw 	& \qw\\
\push{\bit{\ast}}&\push{x_{1}^{i}}&\qw & \ctrld{-1}&\qw &\qw 	&\qw 	&\qw &\qw\\  
\push{\bit{1}}&\push{b_{0}^{i}} &\qw& \ctrl{1}	&\qw &\ctrl{6}&\qw 	&\qw 	& \qw\\
\push{\bit{\ast}}&\push{y_{1}^{i}} &\qw& \ctrld{-1} &\qw &\qw 	&\qw 	&\qw 	&\qw\\  
{\small I_{0}}
&&\push{\bit{1}}&\push{a_{0}^{o}} &\qw	&\qw &\qw 	&\qw 	&\qw&\qw&\push{1}\\
&&\push{\bit{0}}&\push{x_{0}^{o}} &\qw	&\qw &\qw 	&\qw 	 &\qw&\qw& \push{0}\\  
&&\push{\bit{\Z}}&\push{a_{1}^{o}}&\ctrld{-1}&\qw &\qw 	&\qw 	&\qw&\qw& \push{0}\\
&&\push{\bit{0}}&\push{y_{0}^{o}} &\qw	&\qw &\qw	&\qw 	 &\qw&\qw&\push{0}\\  
&&\push{\bit{\Z}}&\push{b_{1}^{o}}&\qw	&\ctrld{-1}&\qw	&\qw 	& \qw&\qw&\push{\ast}\\
&&\push{\bit{\Z}}&\push{x_{1}^{o}}&\qw	&\qw	&\ctrld{-1}&\qw  	&\qw&\qw&\push{1}\\  
&&\push{\bit{1}}&\push{b_{0}^{o}} &\qw	&\qw &\qw 	&\qw 	&\qw&\qw&\push{1}\\
&&\push{\bit{\Z}}&\push{y_{1}^{o}}&\qw	&\qw &\qw 	& \ctrld{-1}&\qw&\qw&\push{\ast} \\
&&&&&&&&&&\\
&&&&&&&&&&{I_{1}}
\gategroup{1}{3}{16}{9}{1em}{-}
\gategroup{1}{1}{8}{1}{1em}{.}
\gategroup{9}{11}{16}{11}{1em}{.}
}
$
\end{minipage}
\caption{}
\label{fig:f2}
\end{figure}

For each man $m$ and woman $w$ in $\calf{I}$ and each pair $j,k$
with $j,k < n$ we say $\Pair(m_{j},w_{k})$ holds iff
$w$ is at the $j^{\rm th}$ position of $m$'s preference list and $m$ is at the $k^{\rm th}$ position of $w$'s preference list.
For each such pair there are two consecutive input wires of $C_\calf{I}$
labelled $m^i_j$ and $w^i_k$ respectively.  (Here the superscript $i$
stands for `input'.)  Hence there are $n^2$ pairs of input wires,
making a total of $2n^2$ input wires.

In addition there are $2n^2$ other wires called `output wires'
labelled in the same order as above; two consecutive wires with
labels $m^o_j$ and $w^o_k$ for each pair satisfying
$\Pair(m_{j},w_{k})$.
These output wires have fixed input values: We let output wire
$m_{0}^{o}$ take input one for every man $m$, and let the rest of
output wires have zero inputs.

The circuit $C_\calf{I}$ has the following comparator gates.  For each pair
$(m^i_j,w^i_k)$ of consecutive inputs there is a gate from wire
$m^i_j$ to $w^i_k$.  After these gates, for every person $p$, we add a
gate from wire $p_{j}^{i}$ to $p_{j+1}^{o}$ for every $j<n-1$.
Note that the order of this last group of wires does not matter.
(See \figref{f2}.)


Given the instance $\calf{I}$ of $\SMP$ with $n$ men and $n$ women, define $\M:\set{0,1,\ast}^{2n^{2}}\rightarrow \set{0,1,\ast}^{2n^{2}}$ to be the function computed by the preceding circuit construction, where the inputs of $\M$ are those fed into the input wires, and the outputs of $\M$ are those produced by the output wires. We will use the following notation. Any sequence $I\in \set{0,1,\ast}^{2n^{2}}$ can be seen as an input of function $\M$, and thus we write $I(p_{j}^{i})$ to denote the input value of wire $p_{j}^{i}$ with respect to $I$. Similarly, if a sequence $J\in \set{0,1,\ast}^{2n^{2}}$ is an output of $\M$, then we write $J(p_{j}^{o})$ to denote the output value of wire $p_{j}^{o}$.

Let sequence $I_{0}\in \set{0,1,\ast}^{2n^{2}}$ be an input of $\M$ defined as follows: $I_{0}(m_{0}^{i})=1$ for every man $m$, and  $I_{0}(w_{0}^{i})=0$ for every woman $w$, and $I_{0}(p_{j}^{i})=\ast$ for every person $p$ and every $j$, $1\le j < n$.  Note that the number of $\ast$'s in the sequence $I_{0}$ is 
\begin{align}
c(n) = 2n^{2}-2n.\label{eq:c}
\end{align}
Our version of Subramanian's method \cite{Sub90,Sub94} consists of computing 
\[I_{c(n)} =\M^{c(n)}(I_{0}),\] 
where $\M^{d}$ simply denotes the $d^{\rm th}$ power of $\M$, i.e. the function we get by composing $\M$ with itself $d$ times. It turns out that
$I_{c(n)}$ is a fixed point of $\M$, i.e. $I_{c(n)} = \M(I_{c(n)})$.
To show this, we define a sequence $I'$ to be an {\em extension} of
a sequence $I$ if $I(p) = I'(p)$ for every person $p$ such that
$I(p) \in \{0,1\}$. In other words, all Boolean values in $I$ are preserved in $J$, even though some $\ast$-values in $I$ might be changed to Boolean values in $J$.  
We can show that $\M(I)$ is an extension of $I$ for
every $I$ which extends $I_0$, and hence $\M^d(I_0)$ extends $I_0$ for
all $d$.  It follows that $\M^{c(n)}(I_0)$ is a fixed point because
there are at most $c(n)$ $\ast$'s to convert to 0 or 1.

Now we can extract a stable marriage from the fixed point $I_{c(n)}$ 
by letting $B$ be the sequence obtained by substituting zeros for all
remaining $\ast$-values in $I_{c(n)}$.  Then $B$ is also a fixed
point of $\M$.  A stable marriage can then be extracted from $B$ by announcing the marriage of a man $m$ and a woman $w$ iff
$\Pair(m_{j},w_{k})$ and $B(m_{j}^{o})=1$ and $B(w_{k}^{o})=0$. Our goal is to formalize the correctness of this method.

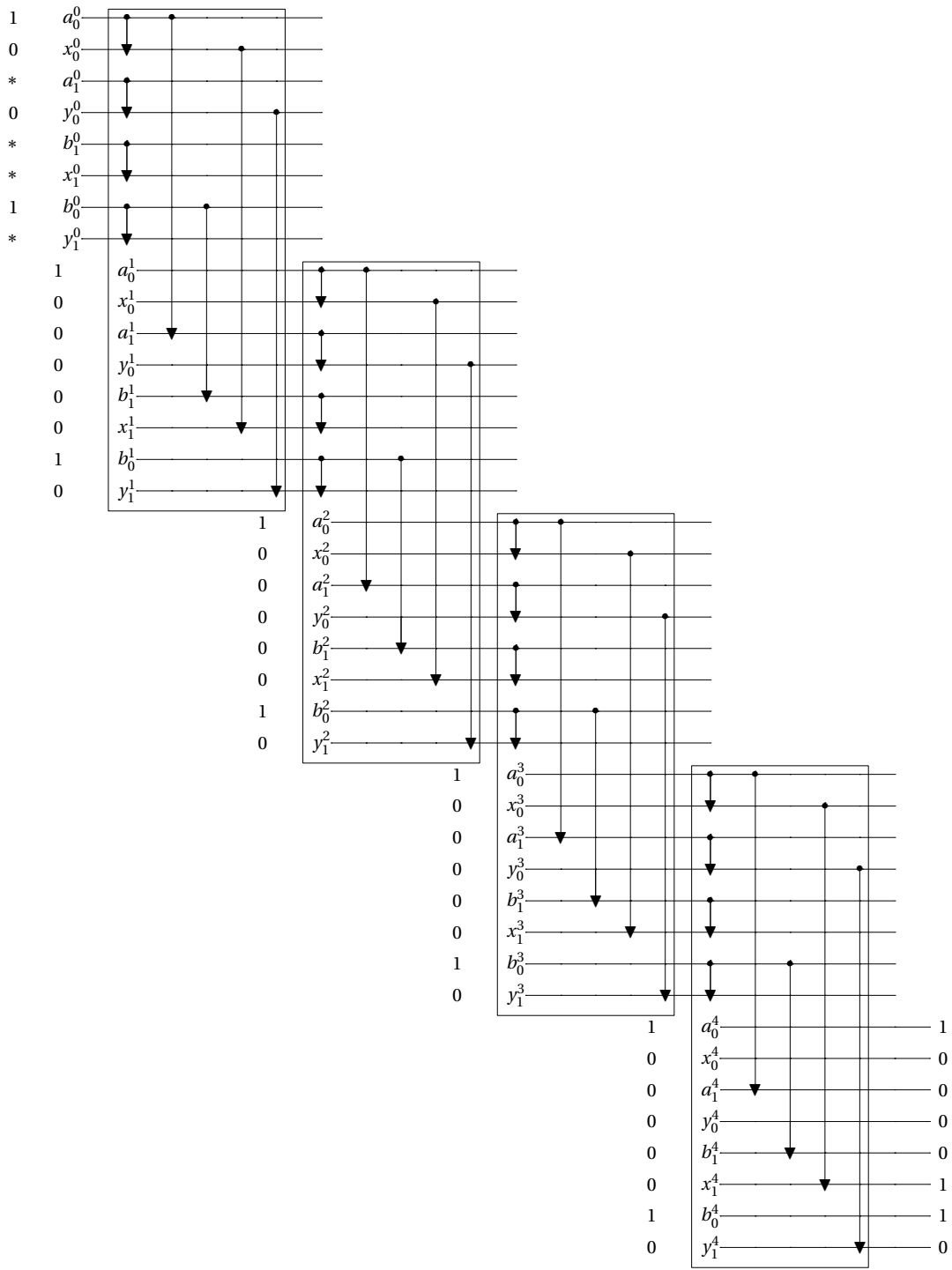
\begin{figure}
\newcommand{\Z}{0}
\newcommand{\one}{\ctrl{9}&\qw 	&\qw 	&\qw 	}
\newcommand{\two}{\qw 	&\qw 	&\ctrl{11}&\qw 	}
\newcommand{\three}{\qw	&\qw 	&\qw 	&\qw 	}
\newcommand{\four}{\qw	&\qw 	&\qw 	&\ctrl{11}}
\newcommand{\five}{\qw	&\qw 	&\qw 	&\qw }
\newcommand{\six}{\qw 	&\qw 	&\qw 	&\qw }	
\newcommand{\sev}{\qw 	&\ctrl{6}&\qw 	&\qw }	
\newcommand{\eight}{\qw 	&\qw 	&\qw 	&\qw 	}

\begin{center}\footnotesize
$\Qcircuit @C=2em @R=0.5em {
\lstick{\bit{1}}&\push{a_{0}^{0}} & \ctrl{1}	&\ctrl{9}&\qw 	&\qw 	&\qw 	& \qw\\
\lstick{\bit{0}}&\push{x_{0}^{0}} & \ctrld{-1}	&\qw 	&\qw 	&\ctrl{11}&\qw 	&\qw\\  
\lstick{\bit{\ast}}&\push{a_{1}^{0}} & \ctrl{1}	&\qw	&\qw 	&\qw 	&\qw 	& \qw\\
\lstick{\bit{0}}&\push{y_{0}^{0}} & \ctrld{-1} 	&\qw	&\qw 	&\qw 	&\ctrl{11}&\qw\\  
\lstick{\bit{\ast}}&\push{b_{1}^{0}} & \ctrl{1}	&\qw	&\qw 	&\qw 	&\qw 	& \qw\\
\lstick{\bit{\ast}}&\push{x_{1}^{0}} & \ctrld{-1}&\qw 	&\qw 	&\qw 	&\qw 	&\qw\\  
\lstick{\bit{1}}&\push{b_{0}^{0}} & \ctrl{1}	&\qw 	&\ctrl{6}&\qw 	&\qw 	& \qw\\
\lstick{\bit{\ast}}&\push{y_{1}^{0}} & \ctrld{-1} &\qw 	&\qw 	&\qw 	&\qw 	&\qw\\  
&\lstick{\bit{1}}&\push{a_{0}^{1}} &\qw	&\qw 	&\qw 	&\qw 	& \ctrl{1}	&\one& \qw\\
&\lstick{\bit{0}}&\push{x_{0}^{1}} &\qw	&\qw 	&\qw 	&\qw 	& \ctrld{-1} 	&\two&\qw\\  
&\lstick{\bit{\Z}}&\push{a_{1}^{1}}&\ctrld{-1}&\qw &\qw 	&\qw 	& \ctrl{1}	&\three& \qw\\
&\lstick{\bit{0}}&\push{y_{0}^{1}} &\qw	&\qw 	&\qw	&\qw 	& \ctrld{-1} 	&\four&\qw\\  
&\lstick{\bit{\Z}}&\push{b_{1}^{1}}&\qw	&\ctrld{-1}&\qw	&\qw 	& \ctrl{1}	&\five& \qw\\
&\lstick{\bit{\Z}}&\push{x_{1}^{1}}&\qw	&\qw	&\ctrld{-1}&\qw & \ctrld{-1} 	&\six&\qw\\  
&\lstick{\bit{1}}&\push{b_{0}^{1}} &\qw	&\qw 	&\qw 	&\qw 	& \ctrl{1}	&\sev& \qw\\
&\lstick{\bit{\Z}}&\push{y_{1}^{1}}&\qw	&\qw 	&\qw 	& \ctrld{-1}& \ctrld{-1}&\eight&\qw\\  
&&&&&&\lstick{\bit{1}}&\push{a_{0}^{2}}	&\qw	&\qw 	&\qw 	&\qw 	& \ctrl{1}	&\one& \qw\\
&&&&&&\lstick{\bit{0}}&\push{x_{0}^{2}}	&\qw	&\qw 	&\qw 	&\qw 	& \ctrld{-1} 	&\two&\qw\\  
&&&&&&\lstick{\bit{\Z}}&\push{a_{1}^{2}}	&\ctrld{-1}&\qw &\qw 	&\qw 	& \ctrl{1}	&\three& \qw\\
&&&&&&\lstick{\bit{0}}&\push{y_{0}^{2}}	&\qw	&\qw 	&\qw	&\qw 	& \ctrld{-1} 	&\four&\qw\\  
&&&&&&\lstick{\bit{\Z}}&\push{b_{1}^{2}}	&\qw	&\ctrld{-1}&\qw	&\qw 	& \ctrl{1}	&\five& \qw\\
&&&&&&\lstick{\bit{\Z}}&\push{x_{1}^{2}}	&\qw	&\qw	&\ctrld{-1}&\qw & \ctrld{-1} 	&\six&\qw\\  
&&&&&&\lstick{\bit{1}}&\push{b_{0}^{2}}	&\qw	&\qw 	&\qw 	&\qw 	& \ctrl{1}	&\sev& \qw\\
&&&&&&\lstick{\bit{\Z}}&\push{y_{1}^{2}}	&\qw	&\qw 	&\qw 	& \ctrld{-1}& \ctrld{-1}&\eight&\qw\\  
&&&&&&&&&&&\lstick{\bit{1}}&\push{a_{0}^{3}}			&\qw	&\qw 	&\qw 	&\qw 	& \ctrl{1}	&\one& \qw\\
&&&&&&&&&&&\lstick{\bit{0}}&\push{x_{0}^{3}}			&\qw	&\qw 	&\qw 	&\qw 	& \ctrld{-1} 	&\two&\qw\\  
&&&&&&&&&&&\lstick{\bit{\Z}}&\push{a_{1}^{3}}		&\ctrld{-1}&\qw &\qw 	&\qw 	& \ctrl{1}	&\three& \qw\\
&&&&&&&&&&&\lstick{\bit{0}}&\push{y_{0}^{3}}			&\qw	&\qw 	&\qw	&\qw 	& \ctrld{-1} 	&\four&\qw\\  
&&&&&&&&&&&\lstick{\bit{\Z}}&\push{b_{1}^{3}}		&\qw	&\ctrld{-1}&\qw	&\qw 	& \ctrl{1}	&\five& \qw\\
&&&&&&&&&&&\lstick{\bit{\Z}}&\push{x_{1}^{3}}		&\qw	&\qw	&\ctrld{-1}&\qw & \ctrld{-1} 	&\six&\qw\\  
&&&&&&&&&&&\lstick{\bit{1}}&\push{b_{0}^{3}}			&\qw	&\qw 	&\qw 	&\qw 	& \ctrl{1}	&\sev& \qw\\
&&&&&&&&&&&\lstick{\bit{\Z}}&\push{y_{1}^{3}}			&\qw	&\qw 	&\qw 	& \ctrld{-1}& \ctrld{-1}&\eight&\qw\\  
&&&&&&&&&&&&&&&&\lstick{\bit{1}}&\push{a_{0}^{4}}		&\qw	&\qw 	&\qw 	&\qw 	& \qw	&\rstick{1}\qw\\
&&&&&&&&&&&&&&&&\lstick{\bit{0}}&\push{x_{0}^{4}} 		&\qw	&\qw 	&\qw 	&\qw 	& \qw	&\rstick{0}\qw\\  
&&&&&&&&&&&&&&&&\lstick{\bit{\Z}}&\push{a_{1}^{4}} 		&\ctrld{-1}&\qw &\qw 	&\qw 	& \qw	&\rstick{0}\qw\\
&&&&&&&&&&&&&&&&\lstick{\bit{0}}&\push{y_{0}^{4}} 		&\qw	&\qw 	&\qw	&\qw 	& \qw	&\rstick{0}\qw\\  
&&&&&&&&&&&&&&&&\lstick{\bit{\Z}}&\push{b_{1}^{4}} 		&\qw	&\ctrld{-1}&\qw	&\qw 	& \qw	&\rstick{0}\qw\\
&&&&&&&&&&&&&&&&\lstick{\bit{\Z}}&\push{x_{1}^{4}} 		&\qw	&\qw	&\ctrld{-1}&\qw & \qw	&\rstick{1}\qw\\  
&&&&&&&&&&&&&&&&\lstick{\bit{1}}&\push{b_{0}^{4}} 		&\qw	&\qw 	&\qw 	&\qw 	& \qw	&\rstick{1}\qw\\
&&&&&&&&&&&&&&&&\lstick{\bit{\Z}}&\push{y_{1}^{4}} 		&\qw	&\qw 	&\qw 	& \ctrld{-1}&\qw&\rstick{0}\qw  
\gategroup{1}{3}{16}{7}{1em}{-}
\gategroup{9}{8}{24}{12}{1em}{-}
\gategroup{17}{13}{32}{17}{1em}{-}
\gategroup{25}{18}{40}{22}{1em}{-}
}
$
\end{center}
\caption{The comparator circuit computing the fourth power of $\M$ constructed from the example in \figref{f2}.  Since the output wires of the
first three blocks serve as input wires to the next block, we use
superscripts 1,2,3,4 for these wires on successive blocks, instead
of the letter `o' used in \figref{f2}.}
\label{fig:f3}
\end{figure}

In the example in \figref{f2}, the computation of the fourth power of the function $\M$ can easily be done by using the comparator circuit in \figref{f3}. We show the outputs, produced as a result of the iteration, which
comprise the fixed point $I_{4}=\M^{4}(I_{0})$ of $\M$.  In this case
the fixed point consists of
Boolean values, where $(I_{4}(a^{o}_{0}),I_{4}(x^{o}_{0}))=(1,0)$ and $(I_{4}(b^{o}_{0}),I_{4}(y^{o}_{1}))=(1,0)$.  Thus, woman $x$ is married to man $a$, and woman $y$ is married to man $b$, which is a stable marriage. \\

Henceforth, we will let $I_{\ell}$ denote the output of $\M^{\ell}(I_{0})$.  We want to show that the $c(n)^{\rm th}$ power of $\M:\set{0,1,\ast}^{2n^{2}}\rightarrow \set{0,1,\ast}^{2n^{2}}$ on input $I_{0}$ defined above in fact produces a fixed point of $\M$.

\begin{theorem}\label{theo:fxp}
($\VCC\proves$) The three-valued sequence $I_{c(n)} = \M^{c(n)}(I_{0})$ is a fixed point of $\M$.
\end{theorem}

To prove this theorem, we need a new definition. The proof of \theoref{fxp} follows from the important observation that  every sequence $I_{q}$ is extended by all sequences $I_{\ell}$, $\ell\ge q$. And thus when going from  $I_{\ell}$ to $I_{\ell+1}$, the only change that can be made is to change some $\ast$-values in  $I_{\ell}$ to Boolean values in $I_{\ell+1}$. But since we started out with at most $c(n)$ $\ast$-values, we will reach a fixed point in at most $c(n)$ steps. Before proving \theoref{fxp}, we need the following lemma, which says that any given sequence $I_{q}$ is extended by all sequences $I_{\ell}$, $\ell\ge q$. The lemma can be formulated as the following $\Sigma_{0}^{B}$ statement. 
\begin{lemma}\label{lem:fxp}
($\VCC\proves$) For every $q\le c(n)+1$, for every person $p$, and for every $k <n$, if  we have $I_{q}(p_{k}^{i})=v\in \set{0,1}$, then $I_{\ell}(p_{k}^{i})=v$ for all $\ell$ satisfying $q \le \ell \le c(n)+1$.
\end{lemma}

\begin{proof}[Proof of \lemref{fxp}] We prove by $\Sigma_{0}^{B}$ induction on $q\le c(n)+1$. The base case ($q=0$) is easy. With respect to $I_{0}$ the only wires having non-$\ast$ values are  $p_{0}^{i}$ for every person $p$. But the output wires $p_0^o$ corresponding to these input values
are fed these same Boolean values as constant inputs, and these wires are
not involved with any comparator gates.  Thus, these values will be preserved for every $I_{\ell}$ with $q \le \ell \le c(n)+1$.

For the induction step, we are given $q$ such that $0<q\le c(n)+1$, and assume that $I_{q}(p_{k}^{i})=v\in \set{0,1}$ for some person $p$ and $k<n$, we want to show that  $I_{\ell}(p_{k}^{i})=v$ for all $\ell$ satisfying $q \le \ell \le c(n)+1$. We will only argue for the case when $p$ is a man $m$ since the case when $p$ is a woman can be argued  similarly. We consider two cases.  We may have $k=0$, in which case, as argued in the base case, we have $I_{\ell}(m_{0}^{i})=1$ for all $\ell$ satisfying $q \le \ell \le c(n)+1$. Otherwise, we have $k\ge 1$, then since $I_{q}(m_{k}^{o})=v\in \set{0,1}$, from how $\M$ was constructed,  the output wire $p_{k}^{o}$ must have got its non-$\ast$ value $v$ from the wire $m_{k-1}^{i}$, which in turn must carry value $v$ before transferring it to wire $m_{k}^{o}$. But then we observe that wire $m_{k-1}^{i}$ is connected to some wire $w_{r}^{i}$ by a comparator gate (i.e. $\Pair(m_{k-1},w_{r})$ holds)  before being connected by a comparator gate to $m_{k}^{o}$. Thus, from the definition of three-valued comparator gate, the value $v$ produced on $m_{k-1}^{i}$ by the gate $\seq{m_{k-1}^{i},w_{r}^{i}}$ only depends on the non-$\ast$ value(s) of either $I_{q-1}(m_{k-1}^{i})$ or $I_{q-1}(w_{r}^{i})$ or both. In any of these cases, by the induction hypothesis, these non-$\ast$ values of $I_{q-1}$ will be preserved in $I_{\ell}$ for all $\ell$ with $q-1 \le \ell \le c(n)+1$. Hence, we will always get $I_{\ell}(m_{k}^{o})=v$ for all $\ell$ satisfying $q \le \ell \le c(n)+1$.
\end{proof}

\begin{proof}[Proof of \theoref{fxp}] Suppose for a contradiction that  for every $\ell \le c(n)$, $I_{\ell}$ is not a fixed point. In other words, $I_{\ell+1}= \M(I_{\ell})\not=I_{\ell}$ for all $\ell \le c(n)$. Is is important to note that, by \lemref{fxp},  when going from $I_{\ell}$ to  $I_{\ell+1}$, we know that $I_{\ell+1}$ extends $I_{\ell}$. Thus, the only change that $\M$ can make at each stage is to switch the $\ast$-values at some positions in $I_{\ell}$ to Boolean values in $I_{\ell+1}$. If $I_{c(n)}$ is not a fixed point, then by utilizing  the counting in $\VTC^{0}$
(Corollary \ref{cor:VTC}), we can show that the number of $\ast$-values that are switched to Boolean values when going from $I_{0}$ to $I_{c(n)+1}$ is at least $c(n)+1$. This is a contradiction, since we started out with only $c(n)$ $\ast$-values in $I_{0}$, and no additional $\ast$-value was supplied during the iterations of $\M$. The final argument, i.e., the number of $\ast$-values never increases, can be proved more formally by $\Sigma_{0}^{B}$-induction on the layers of the comparator circuit computing $\M^{c(n)}(I_{0})$.
\end{proof}

Although \theoref{fxp} gives us a fixed point of $\M$,  this fixed point may still be three-valued and thus does not give us all the information needed to extract a stable marriage. However, every three-valued fixed point can easily  be extended to a Boolean fixed point as follows. Given a three-valued sequence $I$, we let $I[\ast \rightarrow v]$ denote the sequence we get by substituting $v$ for all the $\ast$-values in $I$.

\begin{proposition}\label{prop:fxp0}
($\VCC\proves$) If $I$ is a three-valued fixed point of $\M$, then $I[\ast \rightarrow 0]$  and $I[\ast \rightarrow 1]$ are Boolean fixed points of $\M$.  
\end{proposition}
\begin{proof}
Suppose that $I$ is a three-valued fixed point of $\M$.  Then when
the circuit $C_\calf{I}$ is presented with input $I$ the output is
also $I$.   Without studying the detailed structure of the circuit but
just observing that the gates compute monotone functions, by induction
on the depth of a gate $g$ in the circuit we can compare the values $v$
and $v'$ of $g$ under the two inputs $I$ and $I[\ast \rightarrow 0]$ as
follows: $v' = v$ if $v \in \{0,1\}$ and $v' = 0$ if $v = \ast$.
In particular this is true of the output gates of the circuit.
Since $I$ is a fixed point, the output of $C_\calf{I}$ under input
$I$ is $I$, and so the output under input $I[\ast \rightarrow 0]$
is $I[\ast \rightarrow 0]$. Thus $I[\ast \rightarrow 0]$ is a fixed
point of the circuit.

A similar argument works for $I[\ast \rightarrow 1]$
\end{proof}

To show that the above method for solving $\SMP$ using three-valued comparator circuits is correct, it remains to justify Subramanian's method for extracting a stable marriage from a Boolean fixed point.  Define $G$ to be an $\AC^{0}$-function, i.e., $\Sigma_{0}^{B}$-definable,  which takes as input a Boolean fixed point $B$ of $\M$, and returns a marriage $M$ such that the pair of man $m$ and woman $w$ is in $M$ iff when $j,k$ are chosen such that
$\Pair(m_{j},w_{k})$ holds, then $B(m_{j}^{o})=1$ and $B(w_{k}^{o})=0$.  It is worth noting that since $B=\M(B)$, we have $B(p_{k}^{i})=B(p_{k}^{o})$ for every person $p$ and every $k<n$; however, the superscripts $o$ and $i$ are useful for distinguishing between input and output values of the comparator circuit $C_\calf{I}$ computing $\M$.  From the construction of $G$ and the fixed-point property of $\M$, we have the following theorem.

\begin{theorem}\label{theo:sm1}
($\VCC\proves$) If $B$ is a Boolean fixed point of $\M$ then $M=G(B)$  is a stable marriage.
\end{theorem}
To prove this theorem, we first need to establish the next two lemmas that capture the basic properties of the  comparator circuit computing $\M$.

\begin{lemma}\label{lem:sm1}
($\VCC\proves$)
Let $B$ be any Boolean input to the circuit $C_\calf{I}$.
\begin{enumerate}
\item  For every man $m$ and every $k<n-1$, if $B(m^{i}_{k})=1$ then
$B(m^{o}_{k+1})=0$ iff $B(w^{i}_{j})=0$, where $w^{i}_{j}$
is the wire that satisfies $\Pair(m_{k},w_{j})$.
\item  For every woman $w$ and every $j<n-1$, if $B(w^{i}_{j})=0$ then
$B(w^{o}_{j+1})=1$ iff $B(m^{i}_{k})=1$, where $m^{i}_{k}$
is the wire that satisfies $\Pair(m_{k},w_{j})$.
\end{enumerate}
\end{lemma}
\begin{proof} We will only prove Part~1 since Part~2 can be shown using a dual argument.  For the ($\Leftarrow$) direction, we recall that since $m_{k}^{i}$ was paired with  $w_{j}^{i}$ when constructing $C_\calf{I}$, we have a comparator gate going from  $m_{k}^{i}$ to $w_{j}^{i}$. Thus, since $B(m^{i}_{k})=1$  and $B(w^{i}_{j})=0$, after the comparator gate $\seq{m_{k}^{i},w_{j}^{i}}$, the wire $m^{i}_{k}$  now carries value zero. But since the output wire  $B(m^{o}_{k+1})$ will carry whatever value forwarded from the wire $m^{i}_{k}$, in this case, we have $B(m^{o}_{k+1})=0$.

For the ($\Rightarrow$) direction, from the construction of $C_\calf{I}$, the only way that we can change from $B(m^{i}_{k})=1$ to $B(m^{o}_{k+1})=0$ is by having a comparator gate  $\seq{m_{k}^{i},w_{j}^{i}}$ connecting $m_{k}^{i}$ with some wire $w_{j}^{i}$, and $B(w^{i}_{j})=0$.
\end{proof}

\begin{lemma}\label{lem:sm2}
($\VCC\proves$)
Let $B$ be any Boolean fixed point for $\M$.
\begin{enumerate}
\item  For every man $m$ and every $k<n$, if $B(m^{i}_{k})=1$ and $B(m^{i}_{k+1})=0$, then 
\begin{align*}
B(m^{i}_{0})=\ldots=B(m^{i}_{k})=1, && B(m^{i}_{k+1})=\ldots = B(m^{i}_{n})=0.
\end{align*}
\item  For every woman $w$ and every $j<n$, if $B(w^{i}_{j})=0$ and $B(m^{o}_{j+1})=1$, then  
\begin{align*}
B(w^{i}_{0})=\ldots=B(w^{i}_{j})=0, &&B(w^{i}_{j+1})=\ldots = B(w^{i}_{n})=1.
\end{align*}
\end{enumerate}
\end{lemma}
\begin{proof} We will only prove Part~1 since Part~2 can be proved by a dual argument. Assume that $m$ is a man and $k<n$ is such that $B(m^{i}_{k})=1$ and $B(m^{i}_{k+1})=0$. We will use $\Sigma_{0}^{B}\textit{-MIN}$ to choose the least $k_{0}\ge 0$ satisfying $B(m^{i}_{k_{0}})=1$ and $B(m^{i}_{k_{0}+1})=0$. We can then prove by $\Sigma_{0}^{B}$ induction on $\ell$,  $k_{0}+1\le \ell <n$, that $B(m^{i}_{\ell})=0$. The base case when $\ell = k_{0}+1$ trivially holds. For the induction step, by the construction of $C_\calf{I}$, we observe that when $B(m^{i}_{\ell-1})=0$, then the wire $m^{i}_{\ell-1}$ will always carry value zero. But since $m^{o}_{\ell}$ will receive whatever value forwarded to it from $m^{i}_{\ell-1}$, we get $B(m^{i}_{\ell})=B(m^{o}_{\ell})=0$. Thus, we have just shown that 
\begin{align*}
B(m^{i}_{0})=\ldots=B(m^{i}_{k_{0}})=1, &&B(m^{i}_{k_{0}+1})=\ldots = B(m^{i}_{n-1})=0.
\end{align*}
But this implies that $k_{0}$ is the only subscript at which the elements of the sequence 
\[B(m_{0}^{i}),B(m_{2}^{i}),\ldots,B(m_{n-1}^{i})\]
change their values from one to zero. Thus, we get $k=k_{0}$, and we are done.
\end{proof}

From the above two lemmas, we can show that using $G$ we can extract
from every Boolean fixed point of $\M$ which extends $I_0$ a perfect
matching.
\begin{lemma}\label{lem:sm3}
($\VCC\proves$) If $B$ is a Boolean fixed point of $\M$
then  $M=G(B)$ is a perfect matching between the men and women of $\calf{I}$.
\end{lemma}
\begin{proof} We will only prove that every man is married to a unique woman in $M$ since the claim that every woman is married to a unique man can be shown similarly. Given a man $m$ we want to show that he is married to a unique woman.  Since $B$ is a fixed point we know $B(m_0^i) = B(m_0^o)=1$, and
by Lemma \ref{lem:sm2} the elements of the sequence 
\[B(m_{0}^{i}),B(m_{1}^{i}),\ldots,B(m_{n-1}^{i})\]
can only change their values from one to zero at most once.   Thus
by Lemma \ref{lem:sm1} and the definition of $M$, $m$ can marry at most once.  
It remains for us to show that $m$ indeed gets married.  Suppose to
the contrary that $m$ remains single in $M$. Then 
\[B(m^{i}_{0})=B(m^{i}_{1})=\ldots=B(m^{i}_{n-1})=1.\] 
For every woman $w$ we can choose $k,j<n$ so
that $\Pair(m_{k},w_{j})$ holds, and so by \lemref{sm1} it follows
that $B(w_{j}^{i})=1$. But $B(w^i_0) = B(w^o_0) = 0$, so the elements
of the sequence 
\[B(w_{0}^{i}),B(w_{1}^{i}),\ldots,B(w_{n-1}^{i})\]
must change their values from zero to one at least once, and by
\lemref{sm2} they change their values exactly once. Thus by
\lemref{sm1} and the definition of $M$, every woman is married to exactly
one man.  Since $m$ was excluded, we have $n$ women paired with at most
$n-1$ men, and this contradicts the pigeonhole principle $\PHP(n-1,M)$.
\end{proof}

\begin{proof}[Proof of \theoref{sm1}] By \lemref{sm3}, we know that $M$ is a perfect matching. Thus it only remains to show that $M$ satisfies the stability condition. Suppose not.  Then there exist men $a,b$ and women $x,y$ such that $x$ is married to $a$ and $y$ is married to $b$, but man $a$ prefers $y$ to $x$ and woman $y$ prefers $a$ to $b$. Since $x$ is married to $a$, by how $M$ was constructed and \lemref{sm2}, there are some $k<n$ and $p<n$ such that $\Pair(a_{k},x_{p})$, and 
\begin{align}
&B(a_{0}^{i})= B(a_{1}^{i})=\ldots = B(a_{k}^{i})=1. \label{eq:sm1}
\end{align}
Similarly, since $y$ is married to $b$, there are some $\ell<n$ and $q<n$ such that $\Pair(b_{\ell},y_{q})$, and 
\begin{align}
B(y_{0}^{i})= B(y_{1}^{i})=\ldots = B(y_{q}^{i})=0. \label{eq:sm2}
\end{align}
Now by the definition of $\Pair$ there must be some $s,t<n$ such that $\Pair(a_{s},y_{t})$ holds.  But since man $a$ prefers $y$ to $x$, and woman $y$ prefers $a$ to $b$, we have $s<k$ and $t<q$. Thus from \eref{sm1}  and \eref{sm2}, we get $B(a_{s}^{i})=1$ and $B(y_{t}^{i})=0$ respectively. Hence, $y$ is also married to $a$, and this contradicts \lemref{sm3}.
\end{proof}

Fix a stable marriage instance $\calf{I}$ with $n$ men and $n$ women, and
let $\M$ be the function computed by the comparator circuit $C_\calf{I}$.
Let $\smset$ denote the set of all stable marriages of $\calf{I}$, and
let $\fxpset$ denote the set of all Boolean
fixed points of $\M$ which extend the input $I_0$ defined from $\calf{I}$.
Note that $\smset$ and $\fxpset$ are exponentially large sets, so they
are not really objects of our theories. In other words, we write
$M\in \smset$ to denote that $M$ satisfies a formula asserting the stable
marriage property, and we write  $I \in \fxpset$ to denote that $I$
satisfies a formula asserting that $I$ is a fixed point of $\M$. 
It was proved in \cite{Sub94} that  there is a one-to-one correspondence
between $\smset$ and $\fxpset$, and that the matchings extracted from
$I_{c(n)}[\ast \rightarrow 0]$  and $I_{c(n)}[\ast \rightarrow 1]$ are
man-optimal and woman-optimal respectively. We now show how to formalize
these results.

We define $F:\smset\rightarrow \fxpset$ to be a function that  takes as input a stable marriage $M$ of $\calf{I}$, and outputs a sequence $I\in \set{0,1}^{2n^{2}}$ defined as follows. For every man $m$ and every woman $w$ that are matched in $M$, if $j,k<n$ are subscripts such that $\Pair(m_{j},w_{k})$ holds, then we assign 
\begin{align}
 \mbox{$I(m_{0}^{i}) = \ldots = I(m_{j}^{i}) = 1$ and $I(m_{j+1}^{i}) = \ldots = I(m_{n-1}^{i}) = 0$} \label{eq:mdef} \\
 \mbox{$I(w_{0}^{i}) = \ldots = I(w_{k}^{i}) = 0$ and $I(w_{k+1}^{i}) = \ldots = I(w_{n-1}^{i}) = 1$.} \label{eq:wdef}
\end{align}
From this definition of $F$, we can prove the following lemma. 

\begin{lemma}\label{lem:corres}
($\VCC\proves$)
The function  $F:\smset\rightarrow \fxpset$ is a bijection. 
\end{lemma}
We first need to verify  that the range of $F$ is indeed contained in $\fxpset$.
\begin{lemma}\label{lem:corres1}
($\VCC\proves$)
If $M$ is a stable marriage, then $I=F(M)$ is a fixed point of $\M$.
\end{lemma}
\begin{proof}
We start by stating the following:
\begin{myclaim}
For every pair of wires $(m_{j}^{i},w_{k}^{i})$ satisfying $\Pair(m_{j},w_{k})$,
we have $(I(m_{j}^{i}),I(w_{k}^{i}))=(1,0)$ iff man $m$ is matched to woman $w$
in $M$.
\end{myclaim}
To see that $I$ is a fixed point of $\M$ it suffices to show that equations 
(\ref{eq:mdef}) and (\ref{eq:wdef}) hold with the superscripts $i$ replaced
by $o$.  This follows from the Claim and Lemma \ref{lem:sm1} with $B=I$,
together with the observation that for every man $m$ and woman $w$, the
circuit $C_\calf{I}$ assigns the outputs $m^o_0 = 1$ and $w^o_0 = 0$.

It remains to prove the Claim.  The direction $(\Leftarrow)$ follows
immediately from the definition of $I$.  To prove the direction
$(\Rightarrow)$, suppose that for some man $m$ and woman $v$, 
$\Pair(m_{\ell},v_{s})$ holds for some $\ell,s < n$ and 
$(I(m_{\ell}^{i}),I(v_{s}^{i}))=(1,0)$ but $m$ is not matched to $v$.
Then $m$ is matched to some other woman $w$, and $\Pair(m_j,w_k)$
holds for some $j,k< n$.  Since $I(m_\ell^i)=1$, it follows from
(\ref{eq:mdef}) that $\ell < j$, so $m$ prefers $v$ to $w$.
Since $I(v_{s}^{i}) = 0$, it follows from (\ref{eq:wdef}) applied
to the woman $v$ that $v$ prefers $m$ to the man that she is matched
with.  Therefore the marriage is not stable.
\end{proof}

\begin{proof}[Proof of \lemref{corres}] From \lemref{corres1}, we know that the function $F$ is properly defined. It also follows from how $F$ was defined that two distinct stable marriages will get mapped to distinct fixed points of $\M$, and hence $F$ is injective. It only remains to show that $F$ is surjective. But then it is not hard to check that the function $G$ defined
before Theorem \ref{theo:sm1} is a left inverse of $F$.
\end{proof}

We next want to show that  $G(I_{c(n)}[\ast \rightarrow 0])$  and $G(I_{c(n)}[\ast \rightarrow 1])$ are man-optimal and woman-optimal stable marriages of $\calf{I}$ respectively. A technical difficult is that it might be tricky to compare $G(I_{c(n)}[\ast \rightarrow 0])$  and $G(I_{c(n)}[\ast \rightarrow 1])$ with $G(J)$ for some arbitrary  Boolean fixed point $J$ of $\M$. However the following lemma shows that every Boolean fixed point of $\M$ is an extension of $I_{c(n)}$, which means that it suffices to work with only Boolean fixed-point extensions of $I_{c(n)}$.

\begin{lemma}\label{lem:bext}
($\VCC\proves$) If $J$ is a Boolean fixed point of $\M$, then $J$ extends $I_{\ell}$ for every $\ell \le c(n)$.
\end{lemma}

\begin{proof} We prove by $\Sigma_{0}^{B}$ induction on $\ell\le c(n)$. Base case ($\ell=0$): we have $I_{0}(m_{0}^{o})=1$ for every man $m$ and   $I_{0}(w_{0}^{o})=0$ for every woman $w$. But from how $\M$ was constructed, $\M$ always outputs value one on wire $m_{0}^{o}$ for every man $m$ and zero on wire $w_{0}^{o}$ for every woman $w$. Thus since $J$ is a Boolean fixed point of $\M$, we also have $J(m_{0}^{o})=1$ for every man $m$ and $J(w_{0}^{o})=1$ for every man $w$, and hence $J$ extends $J_{0}$.

For the induction step, we are given $\ell$ such that $0<\ell \le c(n)$, and assume that $J$ extends $I_{\ell}$. We want to show that  for every person $p$ and $k<n$, if $I_{\ell}(p_{k}^{i})=v \in \set{0,1}$, then $J(p_{k}^{i})=v$. We will only argue for the case when $p$ is a man $m$ since the case when $p$ is a woman can be argued  similarly. We consider two cases.  We may have $k=0$, then we can argue as in the base case. Otherwise, we have $k\ge 1$, then since $I_{\ell}(m_{k}^{o})=v\in \set{0,1}$, from how $\M$ was constructed,  the output wire $p_{k}^{o}$ must have got its non-$\ast$ value $v$ from the wire $m_{k-1}^{i}$, which in turn must have carried value $v$ before transferring it to wire $m_{k}^{o}$. But then we observe that wire $m_{k-1}^{i}$ is connected to some wire $w_{r}^{i}$ by a comparator gate (i.e. $\Pair(m_{k-1},w_{r})$ holds)  before being connected by a comparator gate to $m_{k}^{o}$. Thus from the definition of three-valued comparator gate, the value $v$ produced on $m_{k-1}^{i}$ by the gate $\seq{m_{k-1}^{i},w_{r}^{i}}$ only depends on the non-$\ast$ value(s) of either $I_{m-1}(m_{k-1}^{i})$ or $I_{q-1}(w_{r}^{i})$ or both. In any of these cases, since $J$ extends $I_{\ell-1}$ (by the induction hypothesis), these non-$\ast$ values of $I_{q-1}$ will also be contained $J$. But since $J=\M(J)$, we get $J(m_{k}^{o})=v$.
\end{proof}

From \lemref{corres} and \lemref{bext}, we can prove the following theorem.

\begin{theorem}\label{theo:manopt}
($\VCC\proves$) Let $M$ be a stable marriage of the $\SMP$ instance $\calf{I}$. Let 
\begin{align*}
M_{0} = G(I_{c(n)}[\ast \rightarrow 0]) && M_{1}= G(I_{c(n)}[\ast \rightarrow 1]). 
\end{align*}
Then $M_0$ and $M_1$ are
stable marriages, and every man gets a partner in $M_{0}$ no worse than the one he gets in $M$, and every woman gets a partner in $M_{1}$ no worse than the one she gets in $M$. In other words, $M_{0}$ and $M_{1}$ are  the man-optimal and woman-optimal solutions  respectively.
\end{theorem}
\begin{proof} We  only prove that $M_{0}$ is man-optimal since the proof that $M_{1}$ is woman-optimal is similar. From \lemref{corres} and \lemref{bext}, if we let $K = F(M)$, then $K$ is a Boolean fixed point of $\M$ extending the three-valued fixed point $I_{c(n)}$ and $K$ uniquely determines  $M$. Suppose for a contradiction that some man $m$ gets a better partner in $M$ than in $M_{0}$. Let $w$ and $u$ be the women $m$ marries in $M$ and $M_{0}$, and assume that $j,k,\ell,s<n$ are subscripts such that $\Pair(m_{j},w_{k})$ and $\Pair(m_{\ell},u_{s})$ hold. For brevity, let $O=I_{c(n)}[\ast \rightarrow 0]$. Then from how $M$ and $M_{0}$ are constructed, we have $(K(m_{j}^{i}),K(w_{k}^{i})) = (1,0)$ and  $(O(m_{\ell}^{i}),O(u_{s}^{i})) = (1,0)$. Note that we construct $O$ by substituting zeros for all $\ast$-values in $I_{c(n)}$, so we must have $I_{c(n)}(m_{\ell}^{i})=1$ originally. By \lemref{sm2}, we have $O(m_{0}^{i})=\ldots = O(m_{\ell}^{i})=1$. But since $m$ prefers $w$ to $u$, we also have $j<\ell$, and hence $O(m_{j}^{i})=1$. Since we cannot introduce additional ones to $I_{c(n)}$ to get $O$, we also have $I_{c(n)}(m_{j}^{i})=1$. We next show the following claim, which will imply a contradiction since $m$ cannot marry both $w$ and $u$ in 
the stable marriage $M_{0}$.
\begin{myclaim}
We must have $O(w_{k}^{i})=0$.
\end{myclaim}

We cannot have $(I_{c(n)}(m_{j}^{i}),I_{c(n)}(w_{k}^{i})) = (1,0)$; otherwise, $m$ has no choice but to marry $w$ in both $M$ and $M_{0}$ since both $K$ and $O$ are extensions of $I_{c(n)}$. This forces $I_{c(n)}(w_{k}^{i})=\ast$. But then since we must substitute zeros for all $\ast$-values when producing $O$, we have  $O(w_{k}^{i})=0$ .
\end{proof}

\begin{theorem}\label{theo:mosm}
($\VCC\proves$)  $\MOSM$ and $\WOSM$ are $\AC^{0}$-many-one-reducible to $\CCVN$.
\end{theorem}
\begin{proof} We will show only the reduction from $\MOSM$ to $\CCVN$ since the reduction from $\WOSM$ to $\CCVN$ works similarly.

Following the above construction, we can write a $\Sigma_{0}^{B}$-formula defining an $\AC^{0}$ function that takes as input an instance of $\MOSM$ with preference lists for all the men and women, and produces a three-valued comparator circuit that computes the three-valued fixed point $I_{c(n)}=\M^{c(n)}(I_{0})$, and then extracts the man-optimal stable marriage from $I_{c(n)}[\ast\rightarrow 0]$. Although the first step of computing $I_{c(n)}=\M^{c(n)}(I_{0})$ can easily be done as shown in the example from \figref{f3}, the second step of computing $I_{c(n)}[\ast\rightarrow 0]$ from the output  $I_{c(n)}$ and extracting the stable marriage cannot be trivially done using a three-valued comparator circuit. However, we can apply the construction from the proof of  \theoref{tcv} to simulate the three-valued computation of $I_{c(n)}=\M^{c(n)}(I_{0})$ using an instance of $\CCVN$, where we can then utilize the available negation gates, $\wedge$-gates and $\vee$-gates to build the necessary gadget to decide if a designated pair of man and woman are married in the man-optimal stable marriage. The use of negation gates is essential in our construction.

Let $\calf{I}$ be an instance of $\MOSM$, where $(m,w)$ is the designated pair of man and woman. Let $M$ denote the man-optimal stable marriage of $\calf{I}$. We choose $j,k$  such that $\Pair(m_{j},w_{k})$ holds. Then we recall that  $(m,w)\in M$ iff  $I_{c(n)}[\ast\rightarrow 0](m_{j}^{o})=1$ and $I_{c(n)}[\ast\rightarrow 0](w_{k}^{o})=0$. Observe that $I_{c(n)}[\ast\rightarrow 0](m_{j}^{o})=1$ and $I_{c(n)}[\ast\rightarrow 0](w_{k}^{o})=0$ iff 
\begin{align}
\bigl(I_{c(n)}(m_{j}^{o}),I_{c(n)}(w_{k}^{o})\bigr) = (1,0) \;\vee\;  \bigl(I_{c(n)}(m_{j}^{o}),I_{c(n)}(w_{k}^{o})\bigr) = (1,\ast).
\label{eq:mosm.1}
\end{align} 
Let $C$ denote the three-valued circuit computing $I_{c(n)}=\M^{c(n)}(I_{0})$. Then \eref{mosm.1} simply asserts that the wire carrying $I_{c(n)}(m_{j}^{o})$ of $C$ must output 1 and the wire carrying $I_{c(n)}(w_{k}^{o})$ of $C$ must output either $0$ or $\ast$. Let $C'$ be the boolean comparator circuit that simulates the three-valued computation of $C$ using the construction from the proof of  \theoref{tcv}, where now we use a pair of wires in $C'$ to simulate each three-valued wire of $C$.
From \eref{mosm.1} it suffices to check that $I_{c(n)}(m_{j}^{o})$ is
coded by $\seq{1,1}$ and $I_{c(n)}(w_{k}^{o})$ has first component 0
in its code (the two possibilities are $\seq{0,0}$ and $\seq{0,1}$).
This checking is easily done with comparator gates, together with
a negation gate to verify the 0 output.
\end{proof}

\corref{lfm} and Theorems \ref{theo:tcv}, \ref{theo:me} and \ref{theo:mosm} give us the following corollary.

\begin{corollary}($\VCC\proves$)
The ten problems $\MOSM$, $\WOSM$, $\SMP$, $\CCV$, $\CCVN$, $\TCV$, $\TLFMM$,  $\LFMM$,  $\TVLFMM$ and  $\VLFMM$ are all equivalent under many-one $\AC^{0}$-reductions, where the equivalence of $\SMP$ is
with respect to the search problem version of the reduction defined
in Definition~\ref{d:manyOne}.
\end{corollary}
\begin{proof}  \corref{lfm} and \theoref{tcv} show that $\CCV$, $\CCVN$, $\TCV$, $\TLFMM$,  $\LFMM$,  $\TVLFMM$ and  $\VLFMM$  are all equivalent under many-one $\AC^{0}$-reductions.

\theoref{mosm} shows that  $\MOSM$ and $\WOSM$ are
$\AC^{0}$-many-one-reducible to $\TCV$. \theoref{me} also shows that $\TLFMM$ is $\AC^0$-many-one-reducible to $\MOSM$, $\WOSM$, and $\SMP$. Hence,   $\MOSM$, $\WOSM$, and $\SMP$  is equivalent to the above problems under many-one $\AC^0$-reductions. 
\end{proof}

\section{Conclusion and future work}
Our correctness proof of the reduction from $\SMP$ to $\CCV$ is a nice
example showing the utility of three-valued logic for reasoning
about uncertainty.  Since an  instance of $\SMP$ might not have a unique
solution, the fact that the fixed point $I_{c(n)}=\M^{c(n)}(I_{0})$
is three-valued indicates that the construction cannot fully determine how
all the men and women can be matched. Thus, different Boolean fixed-point
extensions of   $I_{c(n)}$ give us different stable marriages. 

It is worth noting that Subramanian's method is not the ``textbook'' method  for solving $\SMP$. The most well-known is the Gale-Shapley algorithm \cite{GS62}. In fact, our original motivation was to formalize the correctness of the Gale-Shapley algorithm, but we do not know how to talk about the computation of the Gale-Shapley algorithm in $\VCC$ due to the fan-out restriction in comparator circuits. Thus, we leave open the question whether $\VCC$ proves the correctness of the Gale-Shapley algorithm.


We believe that $\CC$ deserves more attention, since on the one hand it
contains interesting complete problems, but on the other hand we have
no real evidence (for example based on relativized inclusions) concerning
whether $\CCV$ is complete for $\P$, and if not, whether it is
comparable to $\NC$.   The perfect matching problem  (for bipartite graphs or 
general undirected graphs) shares these same open questions with $\CCV$.  However
several randomized $\NC^2$ algorithms are known for perfect matching
\cite{KUW86,MVV87}, but no randomized $\NC$ algorithm is known for any
$\CC$-complete problem.

Another open question is whether the three $\CCV$ complexity
classes mentioned in (\ref{e:CCs}) coincide, which is equivalent
to asking whether $\CC$ (the closure of $\CCV$ under 
many-one $\AC^0$-reductions) is closed under oracle $\AC^0$-reductions,
or equivalently whether the function class $\FCC$ is closed under
composition.  A possible way to show this would be to show the existence
of universal comparator circuits, but we do not know whether such circuits exist.

The analogous question for standard complexity classes such as
$\TC^0$, $\L$, $\NL$, $\NC$, $\P$ has an affirmative answer.
That is, each class can be defined as the $\AC^0$ many-one closure
of a complete problem, and the result turns out to be also closed
under $\AC^0$ oracle reducibilities.  (A possible exception is the
function class $\#\L$, whose $\AC^0$ oracle closure is the $\#\L$
hierarchy \cite{AO96}. This contains the integer determinant
as a complete problem.)

\bibliographystyle{plain}
\bibliography{wphp}

\begin{thebibliography}{1}

\bibitem{AO96}
E.~Allender and M.~Ogihara.
\newblock {Relationships among {\rm PL}, {\rm \#L}, and the determinant}.
\newblock {\em RAIRO, Theoretical Informatics and Applications}, 30(1):1--21,
  1996.

\bibitem{Bat68}
K.E. Batcher.
\newblock {Sorting networks and their applications}.
\newblock In {\em Proceedings of the AFIPS Spring Joint Computer Conference
  32}, pages 307--314. ACM, 1968.

\bibitem{CN10}
S.~Cook and P.~Nguyen.
\newblock {\em {Logical foundations of proof complexity}}.
\newblock Cambridge University Press, 2010.

\bibitem{GS62}
D.~Gale and L.S. Shapley.
\newblock {College admissions and the stability of marriage}.
\newblock {\em The American Mathematical Monthly}, 69(1):9--15, 1962.

\bibitem{KUW86}
R.M. Karp, E.~Upfal, and A.~Wigderson.
\newblock {Constructing a perfect matching is in random NC}.
\newblock {\em Combinatorica}, 6(1):35--48, 1986.

\bibitem{MS92}
E.W. Mayr and A.~Subramanian.
\newblock {The complexity of circuit value and network stability}.
\newblock {\em Journal of Computer and System Sciences}, 44(2):302--323, 1992.

\bibitem{MVV87}
K.~Mulmuley, U.V. Vazirani, and V.V. Vazirani.
\newblock Matching is as easy as matrix inversion.
\newblock {\em Combinatorica}, 7(1):105--113, 1987.

\bibitem{Sub90}
A.~Subramanian.
\newblock {\em {The computational complexity of the circuit value and network
  stability problems}}.
\newblock PhD thesis, Dept. of Computer Science, Stanford University, 1990.

\bibitem{Sub94}
A.~Subramanian.
\newblock {A new approach to stable matching problems}.
\newblock {\em SIAM Journal on Computing}, 23(4):671--700, 1994.

\end{thebibliography}

\end{document}